\documentclass[sigconf]{acmart}

\usepackage{multirow}
\usepackage{makecell} 
\usepackage{xcolor}
\usepackage{booktabs}
\usepackage{colortbl}
\usepackage{listings}
\usepackage{enumitem}
\usepackage{footmisc}
\usepackage{xspace}

\usepackage{amsthm}

\theoremstyle{definition}
\newtheorem{assumption}{Assumption}

\newcommand{\paratitle}[1]{\noindent\textbf{#1.}\;\xspace}
\newcommand{\paraword}[1]{\noindent\textbf{#1}}

\definecolor{humanblue}{RGB}{107, 174, 214} 
\definecolor{llmred}{RGB}{251, 106, 74}   

\newcommand{\posnum}[1]{\cellcolor{humanblue!20}{#1}}
\newcommand{\negnum}[1]{\cellcolor{llmred!20}{#1}}

\newcommand{\poschange}[1]{\textcolor{humanblue!80!black}{#1}}
\newcommand{\negchange}[1]{\textcolor{llmred!80!black}{#1}}

\AtBeginDocument{%
  }

\copyrightyear{2026}
\acmYear{2026}
\setcopyright{cc}
\setcctype{by}
\acmConference[WSDM '26]{Proceedings of the Nineteenth ACM International Conference on Web Search and Data Mining}{February 22--26, 2026}{Boise, ID, USA}
\acmBooktitle{Proceedings of the Nineteenth ACM International Conference on Web Search and Data Mining (WSDM '26), February 22--26, 2026, Boise, ID, USA}
\acmPrice{}
\acmDOI{10.1145/3773966.3777988}
\acmISBN{979-8-4007-2292-9/2026/02}

\acmSubmissionID{590}

\begin{document}

\title{How Do LLM-Generated Texts Impact Term-Based Retrieval Models?}

\author{Wei Huang}
\affiliation{%
  \institution{State Key Laboratory of AI Safety, Institute of Computing Technology, Chinese Academy of Sciences}
  \city{}
  \country{}
}
\affiliation{%
  \institution{University of Chinese Academy of Sciences}
  \city{Beijing}
  \country{China}
}
\email{huangwei21b@ict.ac.cn}

\author{Keping Bi}
\authornote{Corresponding author.}
\affiliation{%
  \institution{State Key Laboratory of AI Safety, Institute of Computing Technology, Chinese Academy of Sciences}
  \city{}
  \country{}
}
\affiliation{%
  \institution{University of Chinese Academy of Sciences}
  \city{Beijing}
  \country{China}
}
\email{bikeping@ict.ac.cn}

\author{Yinqiong Cai}
\affiliation{%
  \institution{Baidu Inc.}
  \city{Beijing}
  \country{China}
}
\email{caiyinqiong@baidu.com}

\author{Wei Chen}
\affiliation{%
  \institution{State Key Laboratory of AI Safety, Institute of Computing Technology, Chinese Academy of Sciences}
  \city{}
  \country{}
}
\affiliation{%
  \institution{University of Chinese Academy of Sciences}
  \city{Beijing}
  \country{China}
}
\email{chenwei2022@ict.ac.cn}

\author{Jiafeng Guo}
\affiliation{%
  \institution{State Key Laboratory of AI Safety, Institute of Computing Technology, Chinese Academy of Sciences}
  \city{}
  \country{}
}
\affiliation{%
  \institution{University of Chinese Academy of Sciences}
  \city{Beijing}
  \country{China}
}
\email{guojiafeng@ict.ac.cn}

\author{Xueqi Cheng}
\affiliation{%
  \institution{State Key Laboratory of AI Safety, Institute of Computing Technology, Chinese Academy of Sciences}
  \city{}
  \country{}
}
\affiliation{%
  \institution{University of Chinese Academy of Sciences}
  \city{Beijing}
  \country{China}
}
\email{cxq@ict.ac.cn}

\renewcommand{\shortauthors}{Wei Huang et al.}

\begin{abstract}
As more content generated by large language models (LLMs) floods into the Internet, information retrieval (IR) systems now face the challenge of distinguishing and handling a blend of human-authored and machine-generated texts.  
Recent studies suggest that neural retrievers may exhibit a preferential inclination toward LLM-generated content, while classic term-based retrievers like BM25 tend to favor human-written documents. This paper investigates the influence of LLM-generated content on term-based retrieval models, which are valued for their efficiency and robust generalization across domains.  
Our linguistic analysis reveals that LLM-generated texts exhibit smoother high-frequency and steeper low-frequency Zipf slopes, higher term specificity, and greater document-level diversity. These traits are aligned with LLMs being trained to optimize reader experience through diverse and precise expressions. Our study further explores whether term-based retrieval models demonstrate source bias, concluding that these models prioritize documents whose term distributions closely correspond to those of the queries, rather than displaying an inherent source bias. This work provides a foundation for understanding and addressing potential biases in term-based IR systems managing mixed-source content. Code and supplementary material are available at \url{https://github.com/Trustworthy-Information-Access/LLM-Impact-Term-Retrieval}.
\end{abstract}

\ccsdesc[500]{Information systems~Retrieval models and ranking}

\keywords{Source Bias, Mixed-source Retrieval, AIGC, Large Language Models}

\maketitle

\section{Introduction}
The rise of large language models (LLMs) has profoundly reshaped the online information ecosystem. With their growing adoption in applications such as information seeking, content creation, and writing refinement, the volume of LLM-generated content on the web has surged dramatically. According to \cite{schick2020deep}, artificial intelligence-generated content (AIGC) could account for up to 90\% of the web by 2026. While LLM-generated content is often fluent and coherent, it also poses significant risks, including hallucinated misinformation and embedded biases towards certain groups \cite{wenEvaluating2025}. In this age of pervasive LLM influence, information retrieval (IR) systems must contend with a mix of human-authored and machine-generated content. However, the implications of this evolving landscape for retrieval systems remain largely underexplored.

Recent studies have identified a source bias in neural retrievers, which systematically favor LLM-generated content over human-written texts with similar semantics \cite{daiNeural2024, xuInvisible2024}. In contrast, classic term-based retrievers, such as BM25 \cite{robertson1995okapi}, have been observed to prefer human-authored documents. While these empirical findings are intriguing, they lack comprehensive validation and clear explanations. To bridge this gap, this paper investigates the influence of LLM-generated content on term-based retrieval models, which are lightweight, transparent, and highly interpretable. This analysis is significant, as term-based models remain widely used, particularly for corpora with limited relevance annotations.

Since, in term-based retrievers, documents and queries are represented as vocabulary-size vectors, indicating their term distributions, we begin by identifying the linguistic distribution differences between LLM-generated and human-written texts. Using parallel corpora containing human and LLM-generated documents with similar semantics, we analyze their adherence to Zipf's Law \cite{newman2005power,adamic2002zipf}, term specificity, and expression diversity. Our key findings can be summarized as follows:  
1) LLM-generated texts generally exhibit smoother slopes in the high-frequency range of the Zipf curve (representing core vocabulary) and steeper slopes in the low-frequency range (representing extended vocabulary). This suggests that LLMs employ a more diverse core vocabulary for frequent communication but a less diverse extended vocabulary.  
2) Term specificity, as measured by IDF, is typically higher in LLM-generated texts compared to human-written ones.  
3) LLM-generated documents demonstrate greater average document-level diversity than their human-authored counterparts. Additionally, LLMs tend to use a wider array of synonyms to convey the same meaning.  

These observations align with the principle of least effort in linguistics \cite{cancho2003least, zipfHuman1949}, which suggests that language users, both speakers and listeners, gravitate to minimize their effort during communication. Consequently, the slope of the core vocabulary segment in Zipf's curve tends to approach -1, reflecting a balance of competing forces, as illustrated in Figure \ref{fig:zipf_summarization}. Compared to human-written texts, LLMs reduce the cognitive load on listeners or readers by employing more diverse and specific expressions. This is consistent with the fact that LLMs are trained using reinforcement learning from human feedback (RLHF) to better align with human values. 

In exploring the assumption that term-based retrieval models exhibit a source bias towards human-written text, we aim to determine whether this bias is consistently present and identify the factors that may contribute to it. First, we theoretically demonstrate that the preferences of term-based retrievers are influenced by the degree of alignment between the term distributions of queries and documents. Then, we investigate four representative term-based retrievers: TF-IDF, BM25 \cite{robertson1995okapi}, Query Likelihood (QL) \cite{ponteLanguage1998}, and Divergence from Randomness (DFR) \cite{amati2002probabilistic}. 
Unlike existing studies focusing solely on human-written queries, we also explore LLM-rephrased queries, which are becoming an important query type in retrieval-augmented generation (RAG) as LLMs issue them based on their information needs for retrieval~\cite{asaiSelfRAG2023}.

Consistent with existing findings, we observe that for human-written queries, these retrievers favor human-written documents. Conversely, for LLM-rephrased queries, they tend to prefer LLM-generated documents. 
This observation aligns with our theoretical analysis of the relationship between model preference and distribution alignment. It indicates that term-based retrievers do not have inherent source bias, and we can predict model preferences by measuring the distribution alignment between queries and documents.

Our studies shed light on how LLM-generated texts differ from human-written ones from a linguistic perspective and how they impact term-based retrieval models. We hope our work alleviates concerns about source bias in LLM-generated content, particularly for term-based retrieval models, and offers insights for future efforts to mitigate source bias in retrieval systems, ultimately towards trustworthy information access in the LLM era.

\section{Background and Related Work}

\subsection{Characteristics of LLM-generated Text}

Large language models (LLMs) generate text that differs from human authorship along lexical, syntactic, and stylistic axes. Lexically, prior studies have shown that LLM output employs a narrower vocabulary, uses auxiliary verbs more frequently, and contains fewer content words than human writing
\cite{sealsLongform2023,gengChatGPT2024,munoz-ortizContrasting2024}.
Syntactically, generated sentences are usually simpler, embedding fewer complex modifiers and favoring passive constructions
\cite{munoz-ortizContrasting2024}.
From a stylistic perspective, LLM text is highly coherent and grammatically well-formed
\cite{liao2023differentiating,yan2023detection}, yet it expresses a more limited range of affect, particularly negative emotions, than human prose \cite{munoz-ortizContrasting2024}.

Beyond surface‐level features, LLM‐generated text appears to obey the
same scaling regularities that characterize human language, notably
Zipf’s law in word–frequency distributions
\cite{desantisHuman2024,wangLLMGenerated2024}.
While these findings delineate where LLM-generated language converges with or diverges from human production, their explanatory power remains limited, particularly regarding the underlying mechanisms that cause these differences.

In this paper, we analyze the linguistic distribution differences between human and LLM text through the lenses of Zipf’s law, term specificity, and expression diversity. We then interpret these observations through the principle of least effort, aiming to provide new insights into why term-based retrievers treat the two sources of text differently.

\subsection{Term-based Retrieval Models}

Term-based retrieval methods form the foundation of information retrieval research, where documents and queries are represented based on the bag-of-words assumption that treats text as an unordered collection of terms. Among these methods, the Vector Space Model (VSM) introduces a geometric framework where text is represented as sparse vectors in a term space. The model computes relevance through vector similarity metrics, typically using the TF-IDF weighting scheme~\cite{salton1975vector}.

Probabilistic retrieval models estimate the relevance probability between queries and documents, guided by the Probability Ranking Principle. BM25 represents the most widely-adopted model in this category~\cite{robertson1995okapi, robertson2009probabilistic}. In contrast, language modeling approaches reframe the retrieval problem by estimating the probability of generating the query from a document's term distribution, introducing the query likelihood paradigm~\cite{ponteLanguage1998}.

Several other principled approaches have been developed. The Divergence From Randomness (DFR) models term weights through divergence from random term distributions~\cite{amati2002probabilistic}. Axiomatic models formalize retrieval constraints through formal axioms~\cite{fangFormal2004}. Information-based models measure relevance based on the differences in term behavior between document and collection levels~\cite{clinchantInformationbased2010}.

\subsection{Effect of AIGC on Retrieval Systems}

The increasing presence of AIGC has transformed the information ecosystem. Although large language models (LLMs) demonstrate impressive language understanding and generation capabilities, the text they produce may contain factual errors or hallucinations \citep{huangSurvey2025}, embed societal or demographic biases \citep{ayoub2024inherent}. These quality and reliability concerns become more critical because retrieval systems, which mediate user access to mixed-source corpora, can further magnify the impact of problematic content.

One concrete manifestation of this amplification is the so-called source bias. Recent studies have identified a systematic source bias in neural retrievers: they prefer LLM-generated content over semantically equivalent human-written text~\citep{daiNeural2024,daiCocktail2024,xuInvisible2024}. This bias undermines search fairness and result diversity. For neural models, this preference has partly been attributed to perplexity differences between human and machine text~\citep{wangPerplexity2024}. In contrast, term-based retrieval methods show varying ranking preferences~\cite{daiNeural2024}, yet the underlying mechanisms remain unexplored. This knowledge gap is particularly significant as term-based retrieval methods remain widely adopted for large-scale corpus retrieval due to their efficiency and scalability. Our research examines how language distribution differences between human-written and LLM-generated texts influence term-based retrieval behavior, aiming to provide a deeper understanding of what causes apparent source preferences in mixed-source corpus.

\section{Dataset Setup}
Prior work studies source preference by rewriting documents with an LLM while keeping semantics fixed \cite{daiNeural2024,daiCocktail2024}.  
We extend this idea in two directions: (1) rewriting both documents and queries (reflecting realistic RAG scenarios where both inputs may be LLM-generated), and (2) varying the document generation temperature across three settings (0.2, 0.7, 1.0) to evaluate robustness. Controlling meaning while varying these lexical factors allows us to isolate the effect of text source on term-based retrieval. Furthermore, to ensure external validity, we additionally include the HC3-English corpus~\cite{guoHow2023}, which comprises natively produced AIGC.

\paratitle{LLM-Rewritten Parallel Collections}
We select eight human-authored datasets covering a spectrum of formality: informal forum discussions (ANTIQUE~\cite{hashemiANTIQUE2020}, FiQA-2018~\cite{maiaWWW182018}), moderately formal web/wiki content (MS MARCO~\cite{bajajMS2018}, FEVER~\cite{thorneFEVER2018}, NQ~\cite{kwiatkowski2019natural}, Climate-FEVER~\cite{diggelmannCLIMATEFEVER2021}, HotpotQA~\cite{yangHotpotQA2018}), and highly formal academic papers (SCIDOCS~\cite{cohanSPECTER2020}). For each dataset, we construct a parallel LLM-generated collection by rewriting all human documents with Llama-3-8B-Instruct~\cite{grattafioriLlama2024}, following the semantic-preserving protocol proposed in prior work \cite{daiNeural2024,daiCocktail2024}. 
Queries are rewritten with GPT-4~\cite{openaiGPT42024} at temperature 0.2. We intentionally use two different models: Llama-3 offers a favorable cost-quality trade-off for processing large document collections, while the much smaller number of queries makes the use of GPT-4 economically viable, enabling higher-quality reformulations. For both documents and queries, we manually inspected 200 random pairs per dataset and confirmed semantic preservation.  Dataset statistics are listed in Table~\ref{tab:dataset_stat}. 

\paratitle{Native AIGC Evaluation Set}
Synthetic rewrites afford tight control but may not capture the full variability of naturally occurring AI-generated language. To evaluate more realistic conditions, we include the HC3-English corpus~\cite{guoHow2023}, which provides a natural contrast between human and machine-generated responses. This dataset contains 24.3k user questions, 58.5k human-written answers, and 26.9k answers generated by ChatGPT, spanning five domains: general knowledge explanation (from the ELI5 dataset), encyclopedic content (from Wikipedia), open-domain question answering, medicine, and finance. Crucially, each question is paired with both human and AI-authored answers, enabling direct comparison of retrieval behavior across sources under authentic generation conditions. In our setup, each question serves as a query, and each answer paragraph is treated as a candidate relevant document. 

Further details of the rewriting pipeline, prompts, rewritten‐quality checks (e.g., Jaccard overlap, embedding‐based semantic similarity), and detailed HC3 corpus statistics (including per-domain question and answer counts), are in the supplementary material.

\begin{table*}[t]
\centering
\caption{Statistics of all 8 datasets. Avg. D/Q denotes the average number of relevant documents per query. H-Q/L-Q and H-D/L-D represent human/LLM-written queries and documents, respectively. L-D values are reported under different temperature settings (0.2, 0.7, 1.0).}
\label{tab:dataset_stat}
\renewcommand{\arraystretch}{0.9}
\begin{tabular*}{0.9\textwidth}{l @{\extracolsep{\fill}} rrr rrrr rr} 
\toprule
\multirow{3}{*}{Dataset} & \multirow{3}{*}{\# Query} & \multirow{3}{*}{\# Corpus} & \multirow{3}{*}{Avg. D/Q} 
& \multicolumn{6}{c}{Avg. Word Count} \\
\cmidrule(lr){5-10}
& & & & H-Q & L-Q & H-D & L-D$_{0.2}$ & L-D$_{0.7}$ & L-D$_{1.0}$ \\
\midrule
ANTIQUE         & 2,426 & 403,666 & 11.3 & 9.2  & 10.0 & 40.2 & 41.2 & 42.0 & 43.1 \\
FiQA-2018       & 648   & 57,450  & 2.6  & 10.8 & 12.7 & 133.2 & 105.6 & 106.0 & 106.8 \\
\midrule
MS MARCO        & 6,979  & 542,203 & 1.1  & 5.9  & 6.9  & 58.1 & 57.3 & 58.1 & 59.1 \\
NQ              & 3,446  & 104,194 & 1.2  & 9.2  & 9.9  & 86.9 & 83.4 & 84.0 & 84.9 \\
HotpotQA        & 7,405  & 111,107 & 2.0  & 15.7 & 16.0 & 67.9 & 69.4 & 69.8 & 70.6 \\
FEVER           & 6,666  & 114,529 & 1.2  & 8.5  & 9.2  & 113.4 & 89.6 & 89.6 & 89.9 \\
Climate-FEVER   & 1,535  & 101,339 & 3.0  & 20.2 & 20.3 & 99.4 & 81.3 & 81.4 & 81.9 \\
\midrule
SCIDOCS         & 1,000  & 25,259  & 4.7  & 9.4  & 10.1 & 169.7 & 144.9 & 144.2 & 143.7 \\
\bottomrule
\end{tabular*}
\end{table*}

\section{Linguistic Distribution Analysis}
In term-based retrieval models, both queries and documents are represented as vectors of vocabulary size, with each dimension indicating a term weight. Matching is measured between the query and document term distributions. To understand the impact of LLM-generated texts on term-based retrievers, we analyze how their linguistic distributions differ from human-authored texts from three perspectives: Zipf's law, term specificity, and expression diversity. As a preliminary step, we introduce Zipf's law and interpret it based on the principle of least effort~\cite{cancho2003least, zipfHuman1949}.

\subsection{Zipf's Law}

\begin{figure}[t]
    \centering
    \includegraphics[width=0.9\linewidth]{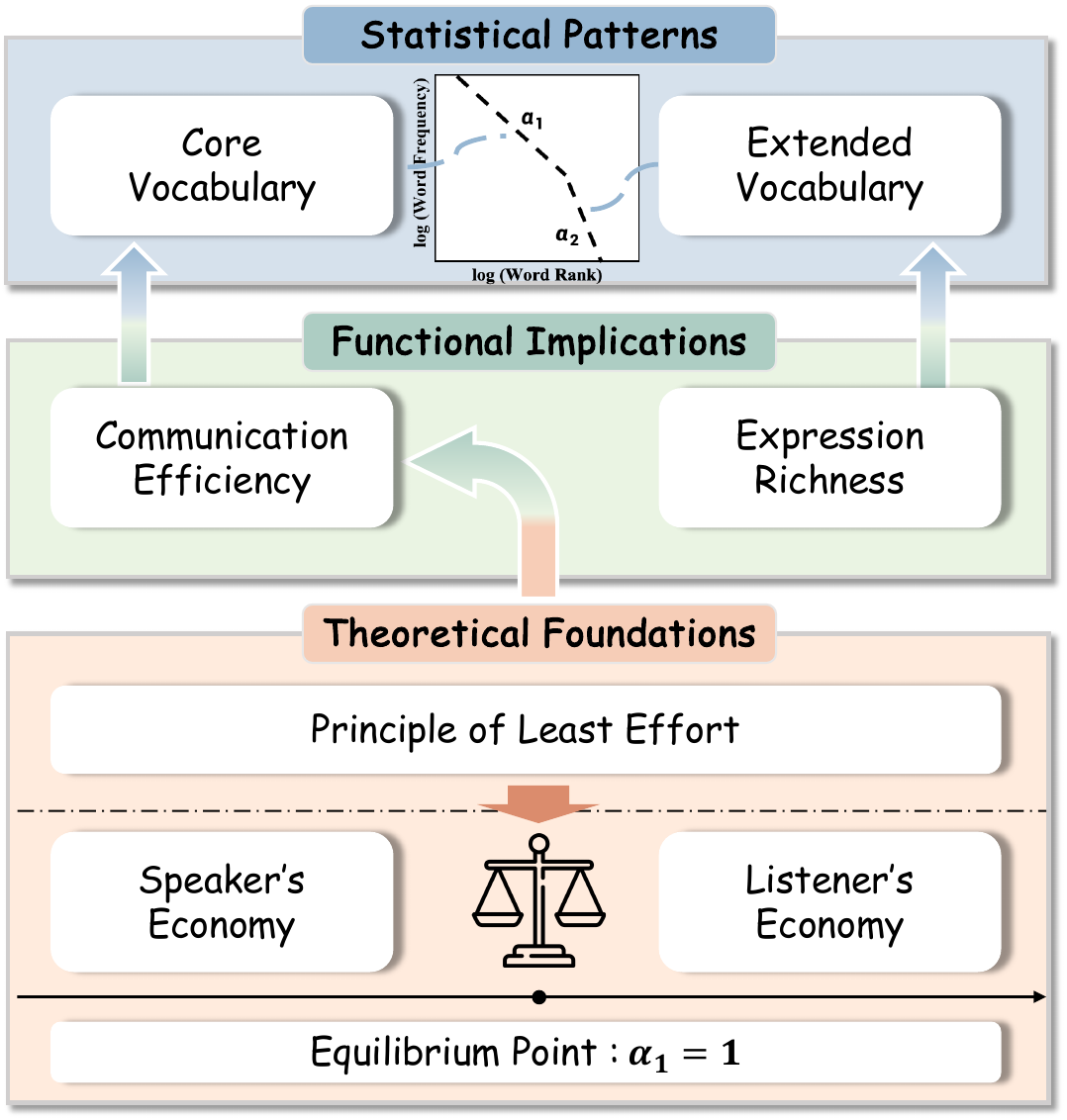}
    \Description{
        A conceptual diagram illustrating a multi-level interpretation of Zipf’s law. The top section presents statistical patterns, distinguishing core vocabulary and extended vocabulary, with a schematic Zipf curve showing two slopes. The middle section shows functional implications, linking communication efficiency and expression richness. The bottom section depicts theoretical foundations based on the principle of least effort, balancing speaker’s economy and listener’s economy, and indicating an equilibrium point. Arrows indicate relationships across the three levels.
    }
    \caption{
        Multi-level interpretation of Zipf's law: from statistical patterns in core and extended vocabulary, to functional implications for communication efficiency and expression richness, to the principle of least effort as a balance of speaker's and listener's economy. 
    }
    \label{fig:zipf_summarization}
\end{figure}

Zipf's Law \cite{newman2005power,adamic2002zipf} is an empirical law observed in natural languages, suggesting word frequency is inversely proportional to its rank.
While prior studies on LLM-generated texts have demonstrated that they adhere to Zipf's Law~\cite{desantisHuman2024,wangLLMGenerated2024}, further interpretations and explanations about the underlying mechanisms are insufficient. We study the traits of LLM-generated texts by examining their adherence to Zipf's Law from three perspectives: statistical patterns, functional implications, and theoretical foundations, as shown in Figure~\ref{fig:zipf_summarization}. Next, we briefly introduce each of them.

\paratitle{Double-Regime Distribution}
When analyzing word frequency against rank on a log-log scale, natural language exhibits a distinctive double-regime pattern~\cite{ferrer2001two}:

\begin{equation}
\!\!\!\log f(r) \propto 
\begin{cases}
-\alpha_1 \log r & \text{for } r \leq r_c \text{ (core vocab.)}, \\
-\alpha_2 \log r & \text{for } r > r_c \text{ (extended vocab.)},
\end{cases}
\end{equation}
where $f(r)$ is the frequency of word with rank $r$, $r_c$ is the transition rank separating core and extended vocabulary, and $\alpha_1$ and $\alpha_2$ are the power-law exponents for the respective vocabulary segments.

\paratitle{Communication Efficiency and Expression Richness}
These two segments reflect distinct aspects of language use. The regime characterized by \(\alpha_1\) corresponds to the core vocabulary, consisting of high-frequency words that facilitate quick and efficient exchange of basic ideas. In contrast, the second regime, defined by \(\alpha_2\), represents the extended vocabulary, comprising low-frequency words that enhance expression richness by enabling nuanced, precise, and context-dependent communication ~\cite{ferrer2001two}.

\paratitle{Principle of Least Effort}
The distribution in the core vocabulary is fundamentally governed by the principle of least effort~\cite{cancho2003least, zipfHuman1949}, which balances two competing forces:

\textsf{Speaker's Economy}: The tendency to minimize expression effort, favoring concentrated word usage, corresponding to a larger $\alpha_1$.

\textsf{Listener's Economy}: The need for clarity and precision, favoring distributed word usage, which is associated with a smaller $\alpha_1$.

In information theory, $\alpha_1 = 1$ represents the theoretical equilibrium between these competing forces~\cite{cancho2003least}. A smaller $\alpha_1$ indicates a preference for the listener's economy through more distributed word usage, while a larger $\alpha_1$ favors the speaker's economy. 

Unlike $\alpha_1$, the parameter $\alpha_2$ in extended vocabulary is not governed by the principle of least effort. Instead, it primarily reflects semantic expressiveness in contexts where precision outweighs efficiency. A smaller $\alpha_2$ indicates richer semantic expression through diverse vocabulary, while a larger value suggests the opposite.

Together, these Zipfian parameters reflect characteristic properties of language production and serve as key metrics to compare human and LLM-generated texts.

\subsection{Zipfian Distribution}
\label{sec:zipfian_distribution}
To empirically validate Zipfian characteristics in LLM-generated texts,we analyze word frequency distributions across nine datasets, including eight LLM-rewritten corpora and the naturally generated HC3-English corpus. For each dataset, we compute word frequencies and ranks separately for human-written and LLM-generated texts. We set the transition rank $r_c = 2000$, motivated by linguistic research showing that 2000-3000 words suffice for basic communication~\cite{nation2006large, adolphs2003lexical}. This choice is validated by our experimental results, showing excellent linear fit in log-log space (R² > 0.99, explaining over 99\% of the data variation) for most datasets. The complete Zipfian parameters for all datasets are presented in Table~\ref{tab:zipf_analysis}.

\begin{table}[t]
\centering
\caption{Zipf's law slopes for human and Llama 3 texts ($\alpha_{1}$: core vocabulary, $\alpha_{2}$: extended vocabulary).
An upward arrow (\poschange{$\uparrow$}) indicates the value for Llama 3 is larger than for the human corpus, while a downward arrow (\negchange{$\downarrow$}) indicates it is smaller.}
\label{tab:zipf_analysis}
\renewcommand{\arraystretch}{0.9}
\begin{tabular}{lcccc}
\toprule
\multirow{2}{*}{Dataset} & \multicolumn{2}{c}{Human} & \multicolumn{2}{c}{Llama 3} \\
\cmidrule(lr){2-3} \cmidrule(lr){4-5}
& $\alpha_1$ & $\alpha_2$ & $\alpha_1$ & $\alpha_2$ \\
\midrule
ANTIQUE & 1.0643 & 1.6479 & 0.9516 \negchange{$\downarrow$} & 1.8511 \poschange{$\uparrow$} \\
FiQA-2018 & 1.0726 & 1.7473 & 0.9766 \negchange{$\downarrow$} & 1.9049 \poschange{$\uparrow$} \\
\midrule
MS MARCO & 0.8873 & 1.6531 & 0.8454 \negchange{$\downarrow$} & 1.7722 \poschange{$\uparrow$} \\
NQ & 0.8657 & 1.4902 & 0.8278 \negchange{$\downarrow$} & 1.5493 \poschange{$\uparrow$} \\
HotpotQA & 0.9211 & 1.4231 & 0.8964 \negchange{$\downarrow$} & 1.4793 \poschange{$\uparrow$} \\
FEVER & 0.9082 & 1.4526 & 0.8892 \negchange{$\downarrow$} & 1.4969 \poschange{$\uparrow$} \\
Climate-FEVER & 0.9037 & 1.4249 & 0.8881 \negchange{$\downarrow$} & 1.4797 \poschange{$\uparrow$} \\
\midrule
SCIDOCS & 0.8852 & 1.7272 & 0.8871 \poschange{$\uparrow$} & 1.7925 \poschange{$\uparrow$} \\\midrule
HC3 & 1.0215 & 1.7003 & 1.0140 \negchange{$\downarrow$} & 1.7889 \poschange{$\uparrow$} \\
\bottomrule
\end{tabular}
\end{table}

\begin{table}[t]
\centering
\caption{Zipf's law slopes for Llama 2 and Qwen2.5. Arrows indicate the change relative to the corresponding human values reported in Table~\ref{tab:zipf_analysis} (\poschange{$\uparrow$} increase, \negchange{$\downarrow$} decrease)}
\label{tab:zipf_cross_model}
\renewcommand{\arraystretch}{0.9}
\begin{tabular}{lcccc}
\toprule
\multirow{2}{*}{Dataset} & \multicolumn{2}{c}{Llama 2} & \multicolumn{2}{c}{Qwen2.5} \\
\cmidrule(lr){2-3} \cmidrule(lr){4-5}
& $\alpha_1$ & $\alpha_2$ & $\alpha_1$ & $\alpha_2$ \\
\midrule
ANTIQUE & 1.0378 \negchange{$\downarrow$} & 1.9683 \poschange{$\uparrow$} & 0.9233 \negchange{$\downarrow$} & 1.9183 \poschange{$\uparrow$} \\

MS MARCO & 0.8622 \negchange{$\downarrow$}  & 1.7670 \poschange{$\uparrow$} & 0.8359 \negchange{$\downarrow$} & 1.7885 \poschange{$\uparrow$} \\
SCIDOCS & 0.9076 \poschange{$\uparrow$} & 1.8071 \poschange{$\uparrow$} & 0.8937 \poschange{$\uparrow$} & 1.8082 \poschange{$\uparrow$} \\
\bottomrule
\end{tabular}
\end{table}

\paratitle{Principle of Least Effort in Human Texts}
Before comparing human and LLM texts, we verify that human corpus themselves follow the principle of least effort, which predicts an equilibrium of $\alpha_{1}=1$ between speaker and listener economy~\cite{cancho2003least}. The results align well with this prediction, as $\alpha_1$ is approximately 1 across most of the corpus. Specifically, informal forums (FiQA, ANTIQUE) show $\alpha_1$ values above 1, favoring the speaker's economy, whereas more formal sources show $\alpha_{1}<1$, emphasizing listener economy. This provides a theoretical background for interpreting the shifts we observe in LLM outputs.

\paratitle{Core Vocabulary Analysis ($\alpha_1$)}
Texts generated by LLMs show a consistently smaller $\alpha_1$ compared to human-written texts in eight out of nine datasets, indicating a more uniform distribution of core vocabulary words. This shift toward listeners' economy likely results from RLHF training~\cite{christianoDeep2017, ouyangTraining2022}, where LLMs are optimized to generate more explicit and detailed content, prioritizing clarity over conciseness. SCIDOCS presents the only exception, where LLM-generated texts show a slightly larger $\alpha_1$, likely because human-written scientific articles naturally favor precise expression.

\paratitle{Extended Vocabulary Analysis ($\alpha_2$)}
In contrast, LLM-generated texts show consistently larger $\alpha_2$ values in the extended vocabulary across all datasets, indicating reduced diversity in specialized and rare word usage. This pattern likely stems from neural networks' inherent limitations in modeling low-frequency words due to insufficient learning.

Concrete examples of these patterns appear in the word‐usage shifts. In the core vocabulary, LLMs systematically replace casual expressions with more formal alternatives, such as substituting ``important'' with ``significant'', ``mostly'' with ``typically''. They also increase the use of discourse markers like ``additionally'' and ``furthermore'', while reducing the frequency of specialized terms (e.g., ``catalyzes'', ``noncoding'') and emotional language (e.g., ``humiliated'').  These lexical patterns further demonstrate LLMs' orientation toward the listener's economy through more diverse core vocabulary while remaining conservative in extended vocabulary usage. The full set of terms exhibiting the largest frequency shifts is available in the supplementary material.

\paratitle{Robustness of the Zipfian Patterns}
To make sure that the Zipf–law shifts we observe are not the result of a
particular generation setting, we conduct two complementary robustness
checks on three representative datasets,
ANTIQUE (informal forum), MS MARCO (web search), and SCIDOCS (scientific papers).

\emph{Sampling temperature.}
For each dataset we regenerate the corpora with three temperatures ($T{=}0.2,\,0.7,\,1.0$).  Figure~\ref{fig:zipf_comparison_temperature}
shows that, for all three collections, LLM texts keep a smaller $\alpha_{1}$ and a larger $\alpha_{2}$ than their human counterparts, regardless of temperature.  Hence the effect cannot be attributed to a specific sampling choice.

\emph{Model family.}
We then set the temperature to $T{=}0.2$ and regenerate the same three datasets with two additional models, Llama-2-7B-Chat~\cite{touvronLlama2023} and
Qwen2.5-7B-Instruct~\cite{yangQwen22024}. Table~\ref{tab:zipf_cross_model} reports the  change in $(\alpha_{1},\alpha_{2})$ relative to human text.  Despite differences in architecture and training data, all models reproduce the same
trend: $\alpha_{1}$ decreases (or marginally increases for SCIDOCS) while $\alpha_{2}$ increases. This confirms that the
Zipfian shift is architecture-agnostic.

Together, these two checks demonstrate that the core/extended-vocabulary patterns are stable across both temperature settings
and model families on datasets with very different writing styles.

\begin{figure}[t]
    \centering
    \includegraphics[width=0.87\linewidth]{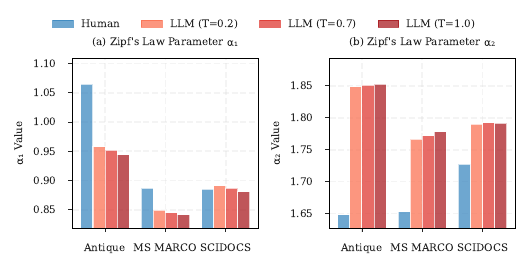}
    \Description{
        Two grouped bar charts comparing Zipf’s-law exponents across datasets and sampling temperatures. The left subplot shows the core-vocabulary exponent $\alpha_{1}$ and the right subplot shows the extended-vocabulary exponent $\alpha_{2}$. Bars correspond to human queries and LLM-generated queries with different sampling temperatures. The horizontal axis lists the datasets (ANTIQUE, MS MARCO, and SCIDOCS), and the vertical axis shows the estimated exponent values.
    }
    \caption{
        Impact of sampling temperature on the two Zipf’s-law exponents. (a) Core-vocabulary slope $\alpha_{1}$; (b) extended-vocabulary slope $\alpha_{2}$.
    }
    \label{fig:zipf_comparison_temperature}
\end{figure}

\paratitle{Zipfian Analysis Summary} 
Our analysis reveals distinct Zipfian patterns that differentiate LLM-generated texts from those written by humans. Specifically, LLMs exhibit smaller $\alpha_1$ values in core vocabulary, showing a preference for listener economy through diverse expressions. Conversely, larger $\alpha_2$ values in extended vocabulary reflect limited specialized term usage. These findings align with the fact that LLMs prioritize reader comprehension through RLHF training while showing limitations in modeling rare words. The consistency across all nine datasets, including both rewritten and naturally generated AIGC, highlights the universality.

\subsection{Term Specificity}

To measure how LLMs alter vocabulary specificity patterns, we analyze the Inverse Document Frequency (IDF) distributions. IDF is a standard term specificity measure~\cite{sparck1972statistical}, where higher values indicate terms appearing in fewer documents and thus carrying more specific information. We compute IDF using a smoothed formula:

\begin{equation}
 \text{IDF} = \text{log}\left(1 + \frac{N-n+0.5}{n+0.5}\right),   
\end{equation}
where $N$ is the total document count and $n$ is the number of documents containing the term. For fair comparison, we use human-written corpus as the reference to compute IDF values for both human and LLM texts. Following our Zipf analysis, we examine core ($r \leq r_c$) and extended ($r > r_c$) vocabulary segments separately.

Our analysis reveals contrasting patterns in different vocabulary segments across all datasets. In the core vocabulary, LLM-generated texts show higher average IDF values than human-written texts, indicating more specific term usage and aligning with our earlier observation of LLMs favoring listener's economy. In the extended vocabulary, however, LLM texts show lower IDF values, suggesting a more limited range of specialized terms. Figure~\ref{fig:idf_distributions} illustrates these patterns using the ANTIQUE dataset as a representative example.

While Zipf analysis reveals term usage patterns through frequency distributions, IDF analysis confirms them through corpus-level term specificity. Both analyses demonstrate the same conclusion: LLMs use more specific terms in core vocabulary but remain conservative in extended vocabulary.

\begin{figure}[t]
    \centering
    \includegraphics[width=0.87\linewidth]{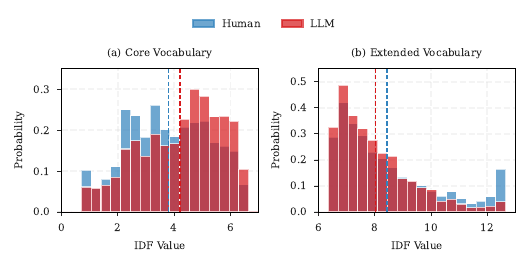}
    \Description{
        Two histogram plots showing inverse document frequency (IDF) distributions in the ANTIQUE dataset. The left subplot presents the core vocabulary and the right subplot presents the extended vocabulary. Blue bars correspond to human queries and red bars correspond to LLM-generated queries. The horizontal axis shows IDF values and the vertical axis shows probability. Vertical dashed lines indicate the mean IDF value for each query type.
    }
    \caption{
        IDF distributions in ANTIQUE. (a) Core vocabulary and (b) extended vocabulary. Histograms represent distributions (blue: Human, red: LLM); vertical dashed lines show respective mean values.
    }
    \label{fig:idf_distributions}
\end{figure}

\subsection{Expression Diversity}

After examining corpus-level term distributions through Zipf's law and IDF analysis, we now focus on more fine-grained diversity patterns at document and synonym levels. These metrics directly capture how different content sources make specific word choices, which ultimately affects term matching in retrieval scenarios.

\begin{figure}[t]
    \centering
    \includegraphics[width=0.87\linewidth]{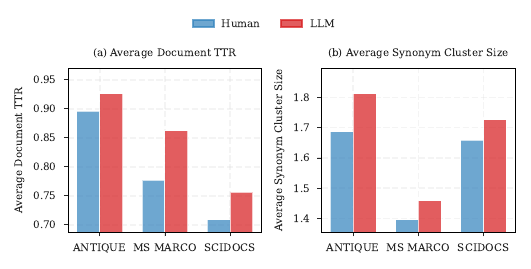}
    \Description{
        Two grouped bar charts comparing lexical diversity between human and LLM queries across three datasets. The left subplot shows average document-level type–token ratio (TTR) values, and the right subplot shows average WordNet-based synonym cluster sizes. Blue bars represent human queries and red bars represent LLM-generated queries. The horizontal axis lists the datasets (ANTIQUE, MS MARCO, and SCIDOCS), and the vertical axis shows the corresponding average values.
    }
    \caption{
       Human vs. LLM lexical diversity. (a) Average document-level TTR values. (b) Average WordNet-based synonym cluster sizes.
    }
    \label{fig:diversity_comparison}
\end{figure}

\paratitle{Document-level Diversity}
We first analyze vocabulary diversity within individual documents using Type-Token Ratio (TTR)~\cite{johnson1944studies, templinCertain1957}. TTR directly quantifies a writer's tendency to reuse vocabulary versus introduce new terms. A higher TTR indicates greater lexical diversity, potentially reflecting more precise word choices. TTR computes the ratio of unique words to total words in a text:

\begin{equation}
\text{TTR} = \frac{\text{number of unique words}}{\text{total number of words}} .
\end{equation}

For each datasets, we compute TTR values for both human-written and LLM-generated texts. Figure~\ref{fig:diversity_comparison}(a) shows representative results from three formality categories, demonstrating that LLM-generated texts consistently have higher TTR values across all datasets. Beyond this measure of vocabulary diversity, we investigate how semantically related terms are distributed. To examine this deeper aspect of vocabulary usage, we turn to synonym analysis.

\paratitle{Synonym Diversity}
To further examine how lexical diversity manifests beyond surface-level variation, we analyze synonym usage using WordNet~\cite{miller1995wordnet}. We measure synonym cluster size, defined as the number of distinct synonyms used to express the same concept (e.g., using both ``small'' and ``tiny'').

As shown in Figure~\ref{fig:diversity_comparison}(b), LLM-generated texts exhibit larger synonym clusters than human-written texts, indicating that LLMs distribute lexical choices across a broader range of semantically related terms. Additional qualitative analyses of representative synonym clusters are provided in the supplementary material.

These analyses at both document and synonym levels reveal consistent patterns: LLMs favor lexical diversity over repetition, and specific terms over general ones. This preference for varied and precise expressions further demonstrates LLMs' orientation toward listener's needs: while humans tend to economize their effort through term repetition and general vocabulary, LLMs appear to emphasize clarity and precision for readers, even at the cost of using a more complex vocabulary set. These lexical diversity patterns align with our earlier finding of smoother Zipf slopes in core vocabulary, providing evidence that LLM texts prioritize listener's economy over speaker's economy.

These systematic differences in vocabulary usage result in two competing effects in retrieval: (1) a more diverse vocabulary increases the probability of matching query terms through various expressions, while (2) it also disperses term frequencies across synonyms, lowering matching scores when exact matches are present. 

Given these competing effects, retrieval outcomes cannot be predicted by analyzing document characteristics in isolation. Our analyses have revealed systematic differences in term distributions between human-written and LLM-generated texts, thus the critical question becomes: how do these term distribution differences between sources affect ranking outcomes in actual retrieval scenarios? In the next section, we address this through a theoretical framework based on distribution alignment, demonstrating how the degree of match between query and document term distributions explains apparent source preferences in ranking results.

\section{Distribution Alignment in Retrieval}

We introduce a theoretical framework centered on the principle of distribution alignment to explain how systematic differences in term distributions translate into ranking results. Under this framework, a retrieval model’s apparent preference for a particular source is not an inherent bias but a direct consequence of how closely the term distribution of a document aligns with that of the query.

To formalize this idea, we begin by specifying the notation and defining the term-based retrieval models under investigation.

\subsection{Notation and Preliminaries}

\subsubsection{Notation}
Let $\text{tf}(w,d)$, $\text{df}(w)$, $|d|$, and $\text{avgdl}$ denote the frequency of term $w$ in document $d$, the document frequency of $w$, the length of $d$, and the average document length in the collection, respectively. Let $V$ be the vocabulary and $N$ be the total number of documents. For probabilistic analysis, $P(w|d)$ denotes the term probability given a document model.

\subsubsection{Term-Based Retrieval Methods}
We analyze four widely-used term-based retrieval methods.

\paraword{TF–IDF} first converts the query and the document into TF–IDF
vectors and then measures their similarity with the cosine function.  
For each term \(w\) we define its TF–IDF weight in a text as:  
\begin{equation}
    t_{x}(w)=\text{tf}(w,x) \cdot \log\frac{N}{\text{df}(w)}, \ \ x\in\{q,d\}.
\end{equation} 

The relevance score of document \(d\) to query \(q\) is the cosine
between their TF–IDF vectors:
\begin{equation}
    \text{Score}_{\text{TF-IDF}}(q,d)=
    \frac{\sum_{w\in V} t_{q}(w)\,t_{d}(w)}
         {
          \sqrt{\sum_{w\in V} t_{q}(w)^{2}}\;
          \sqrt{\sum_{w\in V} t_{d}(w)^{2}} } ,
    \label{eq:tfidf}
\end{equation}

\paraword{BM25} combines inverse document frequency with a
term frequency component that saturates with increasing count:
\begin{align}
    \text{Score}_{\text{BM25}}(q,d) = \sum_{w \in q} & \log \frac{N-\text{df}(w)+0.5}{\text{df}(w)+0.5}  \nonumber \\
    & \cdot \frac{(k_1 + 1) \cdot\text{tf}(w,d)}{k_1((1-b)+b\frac{|d|}{\text{avgdl}})+\text{tf}(w,d)} .
\label{eq:bm25}
\end{align}

\paraword{Query Likelihood (QL)} computes the retrieval score using document generation probability:
\begin{equation}
    \text{Score}_{\text{QL}}(q,d) = \sum_{w \in q} \log P(w|d) .
\end{equation}

\paraword{Divergence From Randomness (DFR)} models the term weight as a product of two information gains: $-log P_1(w|d)$ measures the randomness of term frequency, and $(1-P_2(w|d))$ normalizes the information gain using the term occurrence in the elite set:
\begin{equation}
    \text{Score}_{\text{DFR}}(q,d) = \sum_{w \in q} -\log(P_1(w|d)) \cdot (1-P_2(w|d)).
\end{equation}

\subsection{Distribution Alignment Theorem}

\subsubsection{Proof of the Distribution Alignment Theorem}
\label{sec:theorem_proof}

We first state and prove the theorem for the Query–Likelihood (QL) model, whose score is written purely in probabilistic terms, allowing a concise and transparent proof that illustrates the main idea.

Let $Q$ and $D$ be the sets of queries and documents. We denote a query by $q \in Q$ and a document by $d \in D$. $P_Q(w)$ denotes the marginal probability of term $w$, i.e., its expected occurrence in a query randomly drawn from $Q$. Similarly, $P_D(w)$ denotes the marginal probability of term $w$ in the document collection. With these definitions, we derive the following theorem.

\begin{theorem}
\label{thm:distribuiton_alignment}
For the Query Likelihood method, the expected retrieval score of an arbitrary $q$ and $d$, i.e., $E_{q\in Q,d\in D}[\text{Score}_{\text{QL}}(q,d)]$ has a theoretical upper bound, which is achievable when the document term distribution $P_D$ equals the query term distribution $P_Q$.
\end{theorem}

\begin{proof}
We analyze the expected QL score by taking the expectation over all queries $q \in Q$ and documents $d \in D$: 
\begin{align}
E_{q,d} & [\text{Score}_{\text{QL}}(q,d)] \nonumber \\
&= E_{q,d}[\sum_{w \in q} \log P(w|d)] \nonumber \\
& = E_{q,d}[\sum_{w} \mathbb{I}(w \in q) \log P(w|d)] \nonumber \\
&= \sum_w E_q[\mathbb{I}(w \in q)] \cdot E_d[\log P(w|d)] \quad \text{(Independence)}\nonumber \\
&= \sum_w P_Q(w) \cdot E_d[\log P(w|d)] \quad \text{(Marginal distribution)} \nonumber \\
&\leq \sum_w P_Q(w) \cdot \log E_d[P(w|d)] \quad \text{(Jensen's inequality)} \nonumber \\
&= \sum_w P_Q(w) \cdot \log P_D(w) \nonumber \\ 
&= -\text{KL}(P_Q||P_D) + \sum_w P_Q(w) \cdot \log P_Q(w),
\label{eqn:kl_dive}
\end{align}
where $\mathbb{I}(\cdot)$ is the indicator function. The only source of looseness in the bound comes from Jensen’s inequality. When the document-level term probabilities are similar, this gap disappears. In most cases, the bound is tight.
\end{proof}

\paratitle{Practical implication to source preference}
For two corpus $D$ and $D'$, assume that the query distribution is
closer to $D$ in terms of KL divergence. According to Theorem~\ref{thm:distribuiton_alignment}: 
\begin{equation}
\text{KL}(P_Q||P_D) < \text{KL}(P_Q||P_{D^{'}}) \Rightarrow
  \mathbb{E}_{q,d\in D}[\mathrm{Score}] > \mathbb{E}_{q,d\in D'}[\mathrm{Score}]. \nonumber
\end{equation}
Because ranking is monotonic in the score, documents from $D$ will, on
average, be ranked ahead of those from $D'$.  The system thus
appears to prefer $D$, even though the scoring rule itself is
source-agnostic; the preference simply reflects the tighter
distributional match between $P_D$ and the query distribution~$P_Q$.

\subsubsection{Extension to Other Models}
\label{sec:conceptual_extension}

Having established the distribution alignment principle for the QL model, we now extend it to other term-based models. Complete mathematical proofs for three other classical models (TF–IDF, BM25, and DFR) are given in the supplementary material. Here we only sketch the key ideas.

\paratitle{TF-IDF}
TF–IDF measures relevance via the inner product between idf-weighted query and document vectors. When the corpus term distribution aligns with the query distribution, high-weight terms coincide in both vectors, maximizing their similarity. Formally, by Cauchy–Schwarz inequality, the expected similarity is upper-bounded and attains its maximum when $P_D = P_Q$.

\paratitle{BM25 and DFR}
Intuitively, BM25 and DFR still reward documents containing many query terms, but they diminish the marginal gain of repeated occurrences through a saturating, strictly concave TF--gain (e.g.\ $(k_{1}+1)t/(k_{1}+t)$). This concavity enables an extension of the alignment theorem: taking expectations over queries and documents and applying Jensen’s inequality yields a concave functional $F(P_D)$ on the probability simplex. Its unique optimum, obtained via KKT conditions, has a water-filling form with a negligible length-normalization term of order $O(1/L_d)$. Consequently, the optimal document distribution $P_D$ is practically identical to $P_Q$.

Therefore, all these models are designed to favour documents whose lexical distribution is closely aligned with that of the query. We also empirically validate this unifying principle in Section \ref{sec:experimental_results}.

\subsection{Source Preference Metrics}
\label{sec:position_base_metrics}

Previous work measured retrieval models' preference between human-written and LLM-generated content using relative differences in traditional IR metrics, e.g., NDCG~\cite{daiNeural2024,daiCocktail2024}. However, these metrics are not ideal for this task as they confound a model's preference for a document's source with its judgment of a document's relevance. This can lead to misleading conclusions, particularly when relevance labels are incomplete or when irrelevant documents from one source rank higher than relevant ones from another.

To illustrate this misalignment, consider the example in Table~\ref{tab:preference_misalignment}. In this example, the ranking shows a clear preference for human-written documents, placing them in positions 1, 2, 5, 6, 8, and 10 (6 out of 10 positions). However, when calculating NDCG based on available relevance judgments, the LLM-generated documents would receive a higher score because they occupy positions 3 and 4 with relevant judgments, while many human documents either lack judgments (a common scenario in real-world IR evaluation) or are judged irrelevant. This misleadingly suggests a preference for LLM content despite the actual ranking favoring human documents.

\begin{table}[t]
\footnotesize
\centering
\setlength{\tabcolsep}{4pt}
\caption{Example ranking scenario demonstrating misalignment between source distribution in ranking results and relevance-based measurement which uses Discounted Cumulative Gain (DCG) calculations. Despite humans occupying 6 of 10 positions, relevance-based metrics misleadingly suggest an LLM advantage. Blue cells indicate Human contribution, red cells indicate LLM contribution.}
\label{tab:preference_misalignment}
\begin{tabular}{ccccc}
\toprule
Rank & Source & Relevance & Source-based & Relevance-based \\
\midrule
1 & Human & Unknown & \posnum{+1} & 0 \\
2 & Human & Unknown & \posnum{+1} & 0 \\
3 & LLM & Relevant & \negnum{+1} & \negnum{+0.50} \\
4 & LLM & Relevant & \negnum{+1} & \negnum{+0.43} \\
5 & Human & Irrelevant & \posnum{+1} & 0 \\
6 & Human & Relevant & \posnum{+1} & \posnum{+0.36} \\
7 & LLM & Unknown & \negnum{+1} & 0 \\
8 & Human & Relevant & \posnum{+1} & \posnum{+0.32} \\
9 & LLM & Irrelevant & \negnum{+1} & 0 \\
10 & Human & Irrelevant & \posnum{+1} & 0 \\
\cmidrule(lr){1-5}
\multicolumn{3}{c}{Total contribution} & Human: 6, LLM: 4 & Human: 0.68, LLM: 0.93 \\
\cmidrule(lr){1-5}
\multicolumn{3}{c}{Conclusion} & \posnum{Human advantage} & \negnum{LLM advantage} \\
\bottomrule
\end{tabular}
\end{table}

To focus only on the preference of a document source, we propose three metrics that do not rely on relevance labels:

\paratitle{Source Ratio (SR@k)} measures the proportion of documents from a specific source in top-k results, providing an intuitive measurement similar to precision@k.
\begin{equation}
    \text{SR}_{s}@k = \frac{|\{d_i | \text{source}(d_i) = s, i \leq k\}|}{k},
\end{equation}
where $s \in \{\text{Human}, \text{LLM}\}$ denotes the document source. Similar to precision@k in traditional IR evaluation, SR@k is position-agnostic but offers an intuitive interpretation of source distribution.

\paratitle{Normalized~Discounted~Source~Ratio (NDSR@k)} incorporates position importance by assigning higher weights to top positions, with the same position discounts as NDCG.
\begin{equation}
    \text{NDSR}_{s}@k = \frac{\sum_{i=1}^k w_i \cdot \mathbb{I}(\text{source}(d_i) = s)}{\sum_{i=1}^k w_i} ,
\end{equation}
where $w_i = \frac{1}{\log_2(1+i)}$ assigns higher weights to top positions, and $\mathbb{I}(\cdot)$ is the indicator function. This metric follows the position discount idea from NDCG, ensuring that documents appearing at higher ranks contribute more to the final score.

\paratitle{Mean Average Source Ratio (MASR)} evaluates the entire ranking sequence by accumulating source ratios at each position where a document from a specific source appears, analogous to MAP.
\begin{equation}
    \text{MASR}_{s} = \frac{\sum_{i=1}^n \text{SR}_{s}@i \cdot \mathbb{I}(\text{source}(d_i) = s)}{|\text{docs}_{s}|} ,
\end{equation}
where $\text{SR}_s@i$ is the source ratio at position $i$, and $|\text{docs}_s|$ is the total number of documents from $s$. Analogous to MAP, MASR accumulates source ratios at each position where a document from $s$ appears, providing a comprehensive view of ranking behavior.

\paratitle{Measuring Source Preference} For any source-based metric $M \in \{\mathrm{SR}, \mathrm{NDSR}, \mathrm{MASR}\}$, we measure retrieval preference using a normalized difference:
\begin{equation}
\Delta M
= \frac{M_{\text{human}} - M_{\text{LLM}}}
{M_{\text{human}} + M_{\text{LLM}}}.
\label{eq:relative_preference}
\end{equation}
This normalization yields $\Delta M \in [-1, 1]$, where positive values indicate preference towards human-written documents and negative values indicate preference towards LLM-generated content. As a direct consequence of their construction, SR and NDSR satisfy $M_{\text{human}} + M_{\text{LLM}} = 1$, so $\Delta M$ is equivalent to $M_{\text{human}} - M_{\text{LLM}}$.

\begin{table}[t]
\small
\centering
\caption{Source preference results based on MASR across different datasets and retrieval methods (top 200 per query). 
$\Delta$MASR is computed using the normalized difference defined in
Equation~\ref{eq:relative_preference}. \textcolor{humanblue}{Blue cells} indicate preference towards human-written documents, while \textcolor{llmred}{red cells} indicate preference towards LLM-generated content. * indicates statistical significance ($p < 0.05$).}
\label{tab:preference_analysis_masar_only}
\renewcommand{\arraystretch}{0.9}
\begin{tabular}{cccccccc}
\toprule
\multirow{2}{*}{Dataset} & \multirow{2}{*}{Query Type} & \multicolumn{4}{c}{$\Delta$MASR} \\
\cmidrule(lr){3-6}
& & TF-IDF & BM25 & QL & DFR \\
\midrule
\multirow{2}{*}{ANTIQUE} & Human & \posnum{0.381*} & \posnum{0.366*} & \posnum{0.419*} & \posnum{0.376*} \\
& LLM   & \posnum{0.018*} & \posnum{0.014*} & \posnum{0.040*} & \posnum{0.008} \\
\cmidrule{2-6}
\multirow{2}{*}{FiQA} & Human & \posnum{0.069*} & \posnum{0.179*} & \posnum{0.196*} & \posnum{0.082*} \\
& LLM   & \negnum{-0.110*} & \posnum{0.012} & \posnum{0.015} & \negnum{-0.099*} \\
\midrule
\multirow{2}{*}{MS MARCO} & Human & \posnum{0.361*} & \posnum{0.346*} & \posnum{0.342*} & \posnum{0.351*} \\
& LLM   & \posnum{0.213*} & \posnum{0.194*} & \posnum{0.140*} & \posnum{0.198*} \\
\cmidrule{2-6}
\multirow{2}{*}{NQ} & Human & \posnum{0.125*} & \posnum{0.157*} & \posnum{0.183*} & \posnum{0.127*} \\
& LLM   & \posnum{0.027*} & \posnum{0.055*} & \posnum{0.056*} & \posnum{0.025*} \\
\cmidrule{2-6}
\multirow{2}{*}{HotpotQA} & Human & \posnum{0.124*} & \posnum{0.118*} & \posnum{0.121*} & \posnum{0.118*} \\
& LLM   & \negnum{-0.015*} & \negnum{-0.016*} & \negnum{-0.042*} & \negnum{-0.022*} \\
\cmidrule{2-6}
\multirow{2}{*}{FEVER} & Human & \posnum{0.068*} & \posnum{0.087*} & \posnum{0.099*} & \posnum{0.065*} \\
& LLM   & \negnum{-0.037*} & \negnum{-0.002} & \negnum{-0.021*} & \negnum{-0.040*} \\
\cmidrule{2-6}
\multirow{2}{*}{Climate-FEVER} & Human & \posnum{0.047*} & \posnum{0.120*} & \posnum{0.191*} & \posnum{0.056*} \\
& LLM   & \negnum{-0.033*} & \posnum{0.040*} & \posnum{0.088*} & \negnum{-0.026*} \\
\midrule
\multirow{2}{*}{SciDocs} & Human & \posnum{0.005} & \posnum{0.042*} & \posnum{0.024*} & \posnum{0.008*} \\
& LLM   & \negnum{-0.073*} & \negnum{-0.035*} & \negnum{-0.076*} & \negnum{-0.073*} \\
\bottomrule
\end{tabular}
\end{table}

\begin{figure}[t]
    \centering
    \includegraphics[width=0.90\linewidth]{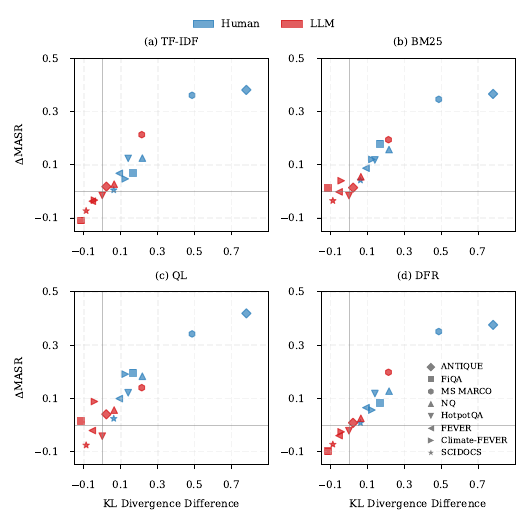}
    \Description{
        Four scatter plots comparing the relationship between KL-divergence difference and changes in MASR for different retrieval methods. Each subplot corresponds to a retrieval model: TF-IDF, BM25, QL, and DFR. The horizontal axis shows the KL-divergence difference between LLaMA-generated and human-written documents, and the vertical axis shows the change in MASR. Blue points represent human queries and red points represent LLM-generated queries, while different marker shapes indicate different datasets.
    }
    \caption{
        KL-divergence differences ($\text{KL}_{\text{LLaMA}} - \text{KL}_{\text{human}}$) versus $\Delta$ MASR for four retrieval methods: (a) TF-IDF (b) BM25, (c) QL, and (d) DFR. Blue and red points represent human and LLM queries respectively, with each shape denoting a different dataset.
    }
    \label{fig:masr_comparison}
\end{figure}

\subsection{Experimental Validation}
\label{sec:experimental_results}
In this section, we analyze source preferences across query and document sources, and provide empirical validation of our distribution alignment theorem. No stemming or stop-word removal is applied, matching the preprocessing used in our Zipf-based analysis.

\paratitle{Setup} MASR is used as the primary metric because it summarizes preference trends over the entire ranked list. Table~\ref{tab:preference_analysis_masar_only} presents MASR scores, and Figure~\ref{fig:masr_comparison} plots the relationship between the KL gap,  $\Delta\mathrm{KL}$ and MASR for four classic term-based retrievers.

\paratitle{Empirical findings} Table~\ref{tab:preference_analysis_masar_only} shows two key empirical findings:
1) With human-written queries, term-based retrievers consistently prefer human-written documents across datasets, confirming the source bias observed in previous studies~\cite{daiNeural2024}. 2) With LLM-generated queries, this preference weakens substantially or even reverses.

\paratitle{Theoretical validation} These patterns can be explained by distributional alignment: Equation~\eqref{eqn:kl_dive} implies that when $\text{KL}(P_Q\|P_D) < \text{KL}(P_Q\|P_{D'})$, documents from $D$ should receive higher scores. Thus, the difference $\Delta\text{KL} = \text{KL}(P_Q\|P_{D'}) - \text{KL}(P_Q\|P_D)$ should correlate positively with retrieval preference.
As shown in Figure~\ref{fig:masr_comparison}, a strong positive correlation exists between $\Delta\mathrm{KL}$ and MASR across all retrieval methods (Pearson coefficients: 0.95, 0.93, 0.96, and 0.95 for BM25, QL, DFR, and TF-IDF, respectively; all with p < 0.001). These results strongly support the distribution alignment theorem,  showing that retrieval preferences are well predicted by distributional alignment between queries and documents.

\paratitle{Additional observations} To further interpret these results, we examine the distribution patterns in Figure~\ref{fig:masr_comparison} in more detail. Human queries span a wider distribution range with stronger preferences for human documents, while LLM queries show compressed ranges with more balanced preferences. This pattern aligns with our linguistic analysis of LLM characteristics. While the alignment principle holds universally, dataset characteristics influence its manifestation. Among the datasets, LLM-generated queries do not consistently favor LLM-generated documents, suggesting that preference may not be determined by the source of the content. For example, MS MARCO maintains a human document preference even with LLM queries, while SciDocs shows dramatic shifts between query types, yet all follow the same alignment principles.

\paratitle{Native-AIGC sanity check} To validate distributional alignment in real-world AIGC, we construct a balanced document collection from HC3-English~\cite{guoHow2023} by sampling one human-written and one ChatGPT-generated answer for each question, then retrieve using original human queries over this mixed-source pool. We observe a consistent preference for human-written answers, with $\Delta\mathrm{MASR}$ averaging $+0.08$ (TF-IDF), $+0.15$ (BM25), $+0.13$ (QL), and $+0.08$ (DFR). The KL divergence from human queries to human answers is smaller than to ChatGPT answers, resulting in a positive $\Delta\mathrm{KL}$. Thus, native AIGC data also align with the distributional alignment principle, supporting the generality of our theorem.

In summary, these consistent alignment patterns across diverse methods and datasets suggest that term-based retrieval models do not possess an inherent source bias; rather, they favor documents whose term distributions better match the query's distribution.

For completeness, we provide additional experiments in the supplementary material, including: (1) results for the other two metrics, which exhibit trends consistent with MASR; (2) analyses using standard preprocessing (stemming and stopword removal); and (3) complete retrieval hyperparameters. All results confirm the robustness and generality of our findings.

\section{Conclusion and Discussion}
\label{sec:conclusion}

This paper presents a systematic investigation of the differences between LLM-generated and human-written texts, and their implications for term-based retrieval systems. Across nine datasets, we find that LLM-generated texts use a more diverse core vocabulary but are more conservative with rare terms—patterns that align with the principle of least effort in language usage. Building on these findings, we theoretically and empirically show that the apparent source bias in term-based retrieval systems is not inherent, but rather a natural consequence of distribution alignment between queries and documents: retrieval models favor documents whose term distributions closely match those of the queries.

These findings have important implications for the field of Information Retrieval and suggest new directions for future research:
\paratitle{Rethinking "Source Bias"}
The term "bias" often carries a negative connotation, implying a flaw within the model. Our work suggests that for term-based IR, "distribution misalignment" is a more accurate and constructive framing. This reframing shifts the focus from attempting to "de-bias" the model itself to managing the dynamic relationship between query and corpus distributions.

\paratitle{Implications for Search in the AIGC Era}
As AIGC becomes dominant, the linguistic distribution of web contents will inevitably shift. Our work predicts that user queries, which still reflect human linguistic patterns, may suffer reduced retrieval effectiveness due to this growing misalignment. This highlights a critical future challenge for search engine design in an AIGC-saturated world.

\paratitle{Future Work and Mitigation Pathways}
Our work suggests several promising avenues for future research. The distribution alignment principle points to concrete mitigation strategies beyond simply filtering AIGC. For instance, systems could pursue:
(a) Query Adaptation: Dynamically rephrasing human queries to better match the stylistic and lexical patterns of the target (potentially AIGC-heavy) corpus.
(b) Alignment-as-a-Feature: Using KL-divergence between a query and a document as a powerful feature in learning-to-rank models to explicitly reward alignment.

\section{Acknowledgment}

This work was funded by the Strategic Priority Research Program of the Chinese Academy of Sciences under Grant No. XDB0680101, the National Key Research and Development Program of China under Grant No. 2023YFA1011602, the CAS Project for Young Scientists in Basic Research under Grant No. YSBR-034, the Xiaomi Young Talents Program/Xiaomi Foundation, the National Natural Science Foundation of China (NSFC) under Grant No. 62302486, the Innovation Funding of ICT CAS under Grant No. E361140, the CAS Special Research Assistant Funding Project, the project under Grant No. JCKY2022130C039, the Strategic Priority Research Program of the CAS under Grant No. XDB0680102, and the NSFC Grant No. 62441229.

\section{Ethics Statement}
Our analysis investigates retrieval performance on LLM-generated versus human-written content. We acknowledge the broader ethical implications associated with AI-generated content, including the spread of misinformation and bias. With these considerations in mind, our research is designed to minimize direct risks by exclusively using publicly available datasets and models that contain no sensitive or personal information. Furthermore, to promote transparency, reproducibility, and facilitate further community-driven research, we release our source code on GitHub.

\bibliographystyle{ACM-Reference-Format}
\balance
\bibliography{wsdm}

\newpage
\appendix

\section{Overview}

This appendix contains the supplementary material referred to throughout the main paper. 
The content is organized into five major sections:

\begin{itemize}[leftmargin=12pt,itemsep=2pt]
    \item \textbf{Section~\ref{appendix:datasets} (Dataset Details)} describes all benchmark datasets used in our study, including construction methodology, rephrasing prompts, corpus statistics, and lexical/semantic analyses. 
    \item \textbf{Section~\ref{appendix:preliminary_experiments} (Preliminary Experiments)} reports additional experiments that motivate our study, focusing on source bias and retrieval preferences under different term-based models. 
    \item \textbf{Section~\ref{appendix:qualitative_lexical} (Qualitative Lexical Differences)} presents concrete word-level examples of the Zipfian shifts, including the words with the largest frequency changes and the usage distribution of the ``important'' synonym cluster. 
    \item \textbf{Section~\ref{appendix:theorem_proof} (Theoretical Proofs)} presents the complete proofs of the distribution alignment theorem for TF-IDF, BM25 and DFR, complementing the discussion in the main article. 
    \item \textbf{Section~\ref{appendix:experiment_setting_and_extension} (Experiment Settings and Extensions)} provides retrieval parameter configurations, results under alternative preprocessing pipelines, and extended evaluations across multiple metrics to validate the robustness of our findings.
\end{itemize}

\section{Dataset Details}
\label{appendix:datasets}

This section provides supplementary details on the datasets used in our experiments, including full descriptions, construction methodology, statistical characteristics, and quality validation.

\subsection{Detailed Description of Datasets} \label{app: datasets_details}
\begin{table*}[t]
    \caption{Website links and licenses for the benchmarked datasets used in this paper.}
    \label{tab:dataset_licenses}
    \begin{tabular}{lll}
        \toprule
        \textbf{Dataset} & \textbf{Official Website Link} & \textbf{License} \\
        \midrule
        ANTIQUE & \url{https://ciir.cs.umass.edu/downloads/Antique/} & Research-only use (no formal license) \\
        FiQA-2018 & \url{https://sites.google.com/view/fiqa/} & Research-only use (no formal license) \\
        MS MARCO & \url{https://microsoft.github.io/msmarco/} & MIT License \\
        NQ & \url{https://ai.google.com/research/NaturalQuestions} & CC BY-SA 3.0 license \\
        HotpotQA & \url{https://hotpotqa.github.io} & CC BY-SA 4.0 license \\
        FEVER & \url{http://fever.ai} & CC BY-SA 3.0 license \\
        Climate-FEVER & \url{http://climatefever.ai} & Research-only use (no formal license) \\
        SCIDOCS & \url{https://allenai.org/data/scidocs} & GNU General Public License v3.0 license \\
        \bottomrule
    \end{tabular}
\end{table*}

\paratitle{ANTIQUE}
ANTIQUE~\cite{hashemiANTIQUE2020} is a non-factoid question answering dataset collected from community forums. It features open-ended questions and long, subjective answers, reflecting the diversity and complexity of real user information needs. This dataset is widely used to evaluate retrieval systems on tasks that require handling nuanced, explanatory, and opinion-based content.

\paratitle{FiQA-2018}
FiQA-2018~\cite{maiaWWW182018} focuses on financial question answering and opinion mining. It consists of real-world questions and community-authored responses drawn from financial forums such as Reddit and Yahoo Finance, supplemented with a smaller set of professional content. The dataset supports research on information retrieval and sentiment analysis in the financial domain, emphasizing domain-specific language and diverse, user-generated information needs.

\paratitle{MS MARCO}
MS MARCO~\cite{bajajMS2018} is a large-scale dataset released by Microsoft, based on anonymized real user queries from Bing search logs. For each query, relevant passages are selected from web documents and manually annotated for relevance. It serves as a benchmark for document ranking, passage retrieval, and machine reading comprehension, and is widely adopted in both information retrieval and natural language processing research.

\paratitle{NQ (Natural Questions)}
The Natural Questions (NQ)~\cite{kwiatkowski2019natural} dataset comprises natural language questions issued by Google users, each paired with relevant Wikipedia documents and annotated answers. It is designed to advance research in open-domain question answering and complex information retrieval, with a focus on modeling authentic user queries.

\paratitle{HotpotQA}
HotpotQA~\cite{yangHotpotQA2018} is a multi-hop question answering dataset that requires reasoning across multiple Wikipedia passages. Each question is accompanied by supporting evidence, challenging models to not only find the correct answer but also to provide interpretable reasoning chains, making it suitable for evaluating inference and evidence aggregation capabilities.

\paratitle{FEVER}
FEVER~\cite{thorneFEVER2018} is a fact verification dataset containing a large collection of claims based on Wikipedia, each labeled with supporting, refuting, or insufficient evidence. It has driven progress in automated fact-checking and evidence retrieval, providing a challenging benchmark for verifying factual statements.

\paratitle{Climate-FEVER}
Climate-FEVER~\cite{diggelmannCLIMATEFEVER2021} is a fact verification dataset tailored to the climate science domain. It includes claims inspired by real-world debates on climate change, each paired with evidence sentences from Wikipedia. The task involves determining whether the evidence supports, refutes, or is insufficient to judge the claim, highlighting the challenges of scientific fact-checking.

\paratitle{SCIDOCS}
SCIDOCS~\cite{cohanSPECTER2020} is a dataset derived from scholarly articles, encompassing citation relationships, topic labels, and metadata. It is used for tasks such as academic document retrieval, citation prediction, and scientific knowledge graph construction, supporting research on information organization and discovery in the scientific literature.

\subsection{Dataset Construction}
\label{app:dataset_construction}

We construct parallel LLM-rephrased versions of both documents and queries for each dataset, following the semantic-preserving rewriting protocol of \citet{daiCocktail2024}. The goal is to create semantically aligned but lexically divergent texts to isolate the effect of text source on retrieval performance. Document rewriting is performed using \text{Llama-3-8B-Instruct} at generation temperatures $T \in \{0.2, 0.7, 1.0\}$, while query rewriting uses \text{GPT-4} at $T=0.2$. To prevent any leakage of relevance or contextual information, each text is rephrased independently based solely on its original content. The original relevance judgments are inherited by the LLM-generated texts under the assumption of semantic equivalence, following standard practice in this setup.

A potential confound in our construction is whether rewriting alters document length systematically. As shown in Table~\ref{tab:app_dataset_stat}, average length remains stable across temperature settings, suggesting that differences in retrieval performance are more likely due to lexical and source characteristics than to length variation.

\subsection{Rephrasing Prompts}
To ensure consistency across document and query rewriting, we employ the same prompt template, shown in Figure~\ref{fig:rephrasing_prompt}.

\begin{figure}[ht]
\begin{lstlisting}[frame=single, basicstyle=\ttfamily\scriptsize, xleftmargin=0pt, numbers=none, breaklines=true]
<|system|>
You are a helpful assistant.

<|user|>
Please follow the instructions below:
1. Maintain the original meaning of the input text.
2. Keep the length of the paraphrased text similar to the original text.
3. Output the paraphrased text directly.

Following is the text you need to paraphrase:
{input text}
Your answer must be formatted as:
Rewritten Text:
<your rewritten text>
\end{lstlisting}
\caption{Prompt used for both document rephrasing with Llama3-8b-Instruct and query rephrasing with GPT-4.}
\label{fig:rephrasing_prompt}
\end{figure}

\begin{table*}[t]
\centering
\caption{Statistics of all 8 datasets. Avg. D/Q denotes the average number of relevant documents per query. H/L-Q and H/L-D represent human/LLM-written queries and documents respectively, with L-D under different temperature settings (0.2/0.7/1.0).}
\label{tab:app_dataset_stat}
\begin{tabular}{cccccccccc}
\toprule
\multirow{3}{*}{Dataset} & \multirow{3}{*}{\# Query} & \multirow{3}{*}{\# Corpus} & \multirow{3}{*}{Avg. D/Q} & \multicolumn{6}{c}{Avg. Word Count} \\
\cmidrule(lr){5-10}
   &          &          &          & H-Q & L-Q & H-D & \makecell{L-D\\Temp0.2} & \makecell{L-D\\Temp0.7} & \makecell{L-D\\Temp1.0} \\
\midrule
ANTIQUE         & 2426      & 403,666    & 11.3      & 9.2          & 10.0     & 40.2       & 41.2      & 42.0      & 43.1  \\
FiQA-2018       & 648       & 57,450     & 2.6       & 10.8         & 12.7     & 133.2      & 105.6     & 106.0     & 106.8  \\
\midrule
MS MARCO        & 6,979     & 542,203    & 1.1       & 5.9          & 6.9      & 58.1       & 57.3      & 58.1      & 59.1  \\
NQ              & 3,446     & 104,194    & 1.2       & 9.2          & 9.9      & 86.9       & 83.4      & 84.0      & 84.9  \\
HotpotQA        & 7,405     & 111,107    & 2.0       & 15.7         & 16.0     & 67.9       & 69.4      & 69.8      & 70.6  \\
FEVER           & 6,666     & 114,529    & 1.2       & 8.5          & 9.2      & 113.4      & 89.6      & 89.6      & 89.9  \\
Climate-FEVER   & 1,535     & 101,339    & 3.0       & 20.2         & 20.3     & 99.4       & 81.3      & 81.4      & 81.9  \\
\midrule
SCIDOCS         & 1,000     & 25,259     & 4.7       & 9.4          & 10.1     & 169.7      & 144.9     & 144.2     & 143.7  \\
\bottomrule
\end{tabular}
\end{table*}

\subsection{HC3-English Corpus Statistics}
\label{app:hc3_stats}

We report basic statistics of the HC3-English corpus~\cite{guoHow2023}, a publicly available dataset containing human and ChatGPT-generated answers to the same questions. The corpus consists of 24,322 unique questions, with a total of 85,449 answer instances: 58,546 human-written responses and 26,903 ChatGPT-generated responses. Each question may have multiple answers, enabling comparative analysis between human and AI-generated text. The distribution of questions and answers across the five source domains is shown in Table~\ref{tab:hc3_domain_stats}. 

\begin{table*}[t]
\centering
\caption{Distribution of questions and answers in the HC3-English corpus, including data sources (from Guo et al., 2023).}
\label{tab:hc3_domain_stats}
\begin{tabular}{lcccl}
\toprule
\textbf{Domain} & \textbf{\# Questions} & \textbf{\# Human Answers} & \textbf{\# ChatGPT Answers} & \textbf{Source} \\
\midrule
reddit\_eli5  & 17,112 & 51,336 & 16,660 & ELI5 dataset~\cite{fanELI52019} \\
open\_qa      & 1,187  & 1,187  & 3,561  & WikiQA dataset~\cite{yangWikiQA2015} \\
wiki\_csai    & 842    & 842    & 842    & Crawled Wikipedia\\
medicine      & 1,248  & 1,248  & 1,337  & Medical Dialog dataset~\cite{zengMedDialog2020} \\
finance       & 3,933  & 3,933  & 4,503  & FiQA dataset~\cite{maiaWWW182018} \\
\midrule
\textbf{Total}& \textbf{24,322} & \textbf{58,546} & \textbf{26,903} & — \\
\bottomrule
\end{tabular}
\end{table*}

\subsection{Lexical and Semantic Analysis}
\label{app:lexical_semantic_analysis}

To assess the quality of our rephrased datasets, we conduct both lexical and semantic analyses following the methodology of \citet{daiCocktail2024}. For lexical analysis, we compute Jaccard similarity and overlap between each LLM-generated document $d^{\text{LLM}}$ and its corresponding human-written counterpart $d^{\text{Human}}$:

\begin{align}
\text{Jaccard similarity} &= \frac{|d^{\text{LLM}} \cap d^{\text{Human}}|}{|d^{\text{LLM}} \cup d^{\text{Human}}|} \\
\text{Overlap} &= \frac{|d^{\text{LLM}} \cap d^{\text{Human}}|}{|d^{\text{Human}}|}
\end{align}

As shown in Table~\ref{tab:jaccard_overlap_similarity}, lexical overlap remains moderate across all datasets and temperature settings, with Jaccard scores below 0.6 and overlap below 0.75, indicating substantial divergence in word choice despite semantic preservation. This property is essential for isolating the effect of text source in retrieval.

For semantic analysis, we measure cosine similarity between document embeddings using OpenAI's \texttt{text-embedding-ada-002} model. Table~\ref{tab:llama3_corpus_similarity} reports similarities for matching pairs (i.e., human and its LLM-rephrased version) and random pairs (i.e., human document paired with a randomly selected LLM-generated document). Matching pairs exhibit consistently high similarity (mean > 0.93 across datasets), indicating strong semantic preservation. In contrast, random pairs yield significantly lower scores (typically below 0.75), confirming that the high similarity is specific to semantically aligned pairs rather than a general property of the embedding space.

These findings confirm that our corpus captures systematic lexical differences attributable to text source while preserving semantics, thereby enabling controlled studies of their effects on retrieval.

\begin{table*}[t]
\centering
\caption{Jaccard and Overlap similarity (mean $\pm$ std) across datasets and LLM temperatures.}
\label{tab:jaccard_overlap_similarity}
\begin{tabular}{lcccccc}
\toprule
\multirow{2}{*}{Dataset} & \multicolumn{3}{c}{Jaccard Similarity} & \multicolumn{3}{c}{Overlap Similarity} \\
\cmidrule(lr){2-4} \cmidrule(lr){5-7}
 & T=0.2 & T=0.7 & T=1.0 & T=0.2 & T=0.7 & T=1.0 \\
\midrule
ANTIQUE & 0.325 ($\pm$0.107) & 0.312 ($\pm$0.101) & 0.297 ($\pm$0.094) & 0.518 ($\pm$0.158) & 0.506 ($\pm$0.157) & 0.493 ($\pm$0.157) \\
FiQA & 0.374 ($\pm$0.093) & 0.357 ($\pm$0.088) & 0.335 ($\pm$0.082) & 0.534 ($\pm$0.107) & 0.519 ($\pm$0.105) & 0.500 ($\pm$0.104) \\
\midrule
MS MARCO & 0.450 ($\pm$0.111) & 0.432 ($\pm$0.107) & 0.410 ($\pm$0.104) & 0.643 ($\pm$0.122) & 0.632 ($\pm$0.122) & 0.616 ($\pm$0.122) \\
NQ & 0.494 ($\pm$0.109) & 0.472 ($\pm$0.106) & 0.445 ($\pm$0.103) & 0.663 ($\pm$0.109) & 0.648 ($\pm$0.110) & 0.628 ($\pm$0.112) \\
HotpotQA & 0.580 ($\pm$0.104) & 0.561 ($\pm$0.103) & 0.536 ($\pm$0.102) & 0.763 ($\pm$0.099) & 0.751 ($\pm$0.101) & 0.736 ($\pm$0.103) \\
FEVER & 0.590 ($\pm$0.106) & 0.569 ($\pm$0.105) & 0.544 ($\pm$0.104) & 0.763 ($\pm$0.108) & 0.751 ($\pm$0.110) & 0.735 ($\pm$0.113) \\
Climate-FEVER & 0.588 ($\pm$0.106) & 0.569 ($\pm$0.105) & 0.543 ($\pm$0.104) & 0.768 ($\pm$0.106) & 0.756 ($\pm$0.108) & 0.740 ($\pm$0.112) \\
\midrule
SCIDOCS & 0.526 ($\pm$0.100) & 0.499 ($\pm$0.095) & 0.464 ($\pm$0.090) & 0.673 ($\pm$0.099) & 0.652 ($\pm$0.098) & 0.624 ($\pm$0.098) \\
\bottomrule
\end{tabular}
\end{table*}

\begin{table*}[t]
\centering
\caption{Cosine similarity (mean $\pm$ std) for matching and random pairs across datasets and LLM temperatures.}
\label{tab:llama3_corpus_similarity}
\begin{tabular}{lcccccc}
\toprule
\multirow{2}{*}{Dataset} & \multicolumn{3}{c}{Matching Pairs} & \multicolumn{3}{c}{Random Pairs} \\
\cmidrule(lr){2-4} \cmidrule(lr){5-7}
 & T=0.2 & T=0.7 & T=1.0 & T=0.2 & T=0.7 & T=1.0 \\
\midrule
ANTIQUE & 0.936 ($\pm$0.029) & 0.933 ($\pm$0.029) & 0.930 ($\pm$0.029) & 0.704 ($\pm$0.027) & 0.704 ($\pm$0.027) & 0.704 ($\pm$0.027) \\
FiQA & 0.960 ($\pm$0.019) & 0.958 ($\pm$0.020) & 0.955 ($\pm$0.021) & 0.721 ($\pm$0.032) & 0.722 ($\pm$0.032) & 0.722 ($\pm$0.032) \\
\midrule
MS MARCO & 0.970 ($\pm$0.021) & 0.969 ($\pm$0.022) & 0.966 ($\pm$0.022) & 0.675 ($\pm$0.028) & 0.675 ($\pm$0.028) & 0.675 ($\pm$0.028) \\
NQ & 0.977 ($\pm$0.018) & 0.976 ($\pm$0.019) & 0.973 ($\pm$0.019) & 0.698 ($\pm$0.028) & 0.698 ($\pm$0.028) & 0.699 ($\pm$0.028) \\
HotpotQA & 0.981 ($\pm$0.016) & 0.979 ($\pm$0.016) & 0.977 ($\pm$0.018) & 0.705 ($\pm$0.028) & 0.705 ($\pm$0.028) & 0.705 ($\pm$0.028) \\
FEVER & 0.979 ($\pm$0.014) & 0.977 ($\pm$0.014) & 0.975 ($\pm$0.015) & 0.696 ($\pm$0.029) & 0.696 ($\pm$0.028) & 0.696 ($\pm$0.028) \\
Climate-FEVER & 0.979 ($\pm$0.014) & 0.977 ($\pm$0.015) & 0.975 ($\pm$0.017) & 0.696 ($\pm$0.028) & 0.696 ($\pm$0.028) & 0.696 ($\pm$0.028) \\
\midrule
SCIDOCS & 0.984 ($\pm$0.014) & 0.982 ($\pm$0.015) & 0.979 ($\pm$0.016) & 0.742 ($\pm$0.031) & 0.743 ($\pm$0.031) & 0.743 ($\pm$0.031) \\
\bottomrule
\end{tabular}%
\end{table*}

\section{Preliminary Experiments}
\label{appendix:preliminary_experiments}

This appendix presents detailed information on the preliminary experiments conducted to assess \textit{source bias}, which refers to the tendency of term-based retrieval models to favor human-written documents when processing human-authored queries. We describe the experimental setup, evaluation methodology, and key results used to analyze retrieval model preferences across different content sources.

\subsection{Experimental Setup}
For all term-based retrieval experiments, we applied consistent preprocessing across datasets, following standard practices in lexical retrieval. Both documents and queries were tokenized using Lucene's analyzer, with Porter stemming for word normalization and removal of stopwords to focus on content-bearing terms. The following parameter settings were used for each retrieval model:

\begin{itemize}
    \item \textbf{TF-IDF}: Standard implementation with logarithmic term frequency and inverse document frequency.
    \item \textbf{BM25}: $k_1=0.9$, $b=0.4$, as recommended in~\cite{yang2018anserini}.
    \item \textbf{Query Likelihood (QL)}: Jelinek-Mercer smoothing with $\lambda=0.1$.
    \item \textbf{DFR}: Inverse document frequency model (In) with Laplace after-effect (L) and hypergeometric normalization (H2), $c=1.0$, $z=0.3$.
\end{itemize}

All models were implemented using the Pyserini toolkit~\cite{linPyserini2021}, extended to support DFR and TF-IDF functionalities.

\subsection{Evaluation Methodology}

To quantify retrieval preferences between different sources, we adopt the methodology from prior work~\cite{daiNeural2024,daiCocktail2024}. Specifically, we evaluate retrieval metrics separately for each content source while preserving the original mixed-source ranking. When evaluating one source (e.g., human-written), documents from the other source (e.g., LLM-generated) are treated as irrelevant. We measure the relative performance difference using the following formula:

\begin{equation}
\text{Relative $\Delta$} = \frac{\texttt{RelMetric}_\text{Human} - \texttt{RelMetric}_\text{LLM}}{(\texttt{RelMetric}_\text{Human} +\texttt{RelMetric}_\text{LLM}) / 2},
\label{eq:source_bias_eval}
\end{equation}
where $\texttt{RelMetric}_\text{Human}$ and $\texttt{RelMetric}_\text{LLM}$ denote relevance scores (e.g., NDCG, MAP, or Precision) computed over human-written and LLM-generated documents, respectively. A positive value indicates a preference for human-written content, while a negative value indicates the opposite tendency. Following prior work, the denominator includes an additional factor of $1/2$, so the resulting relative difference is bounded in $[-2, 2]$ rather than being normalized to a unit interval.

\subsection{Preliminary Results}

Our preliminary experiments consistently showed that term-based retrievers prefer human-written documents when using human-written queries across all settings, confirming and extending previous findings\cite{daiNeural2024}. Table~\ref{tab:preference_analysis} demonstrates this pattern across various retrieval metrics.

\begin{table*}[t]
\centering
\caption{Preference analysis results across different datasets, temperatures, and retrieval methods. 
Temperature (Temp.) refers to the sampling temperature used for generating LLM-generated documents. For all retrieval methods, we retrieved the top 200 documents. \textcolor{humanblue}{blue cells} indicate preference towards human-written documents, while \textcolor{llmred}{red cells} indicate preference towards LLM-generated content.  $\Delta$P@10, $\Delta$NDCG@10, and $\Delta$MAP represent the preference metrics based on Precision@10, NDCG@10, and MAP respectively. * indicates statistical significance ($p < 0.05$).}
\label{tab:preference_analysis}
\aboverulesep=0pt
\belowrulesep=0pt
\renewcommand{\arraystretch}{1.3}
\begin{tabular}{cccccccccccccc}
\toprule
\multirow{2}{*}{Dataset} & \multirow{2}{*}{Temp.} & \multicolumn{4}{c}{$\Delta$P@10} & \multicolumn{4}{c}{$\Delta$NDCG@10} & \multicolumn{4}{c}{$\Delta$MAP} \\
\cmidrule(lr){3-6} \cmidrule(lr){7-10} \cmidrule(lr){11-14}
& & TF-IDF & BM25 & QL & DFR & TF-IDF & BM25 & QL & DFR & TF-IDF & BM25 & QL & DFR \\
\midrule
\multirow{3}{*}{ANTIQUE} & 0.2 & \posnum{0.629*} & \posnum{0.530*} & \posnum{0.634*} & \posnum{0.604*} & \posnum{0.696*} & \posnum{0.601*} & \posnum{0.719*} & \posnum{0.678*} & \posnum{0.730*} & \posnum{0.627*} & \posnum{0.736*} & \posnum{0.713*} \\
& 0.7 & \posnum{0.630*} & \posnum{0.523*} & \posnum{0.635*} & \posnum{0.623*} & \posnum{0.696*} & \posnum{0.591*} & \posnum{0.690*} & \posnum{0.687*} & \posnum{0.747*} & \posnum{0.640*} & \posnum{0.746*} & \posnum{0.738*} \\
& 1.0 & \posnum{0.700*} & \posnum{0.574*} & \posnum{0.698*} & \posnum{0.693*} & \posnum{0.769*} & \posnum{0.650*} & \posnum{0.757*} & \posnum{0.759*} & \posnum{0.792*} & \posnum{0.696*} & \posnum{0.790*} & \posnum{0.787*} \\
\cmidrule(l){2-14}
\multirow{3}{*}{FiQA-2018} & 0.2 & \posnum{0.082} & \posnum{0.168*} & \posnum{0.172*} & \posnum{0.109*} & \posnum{0.091*} & \posnum{0.223*} & \posnum{0.243*} & \posnum{0.112*} & \posnum{0.112*} & \posnum{0.289*} & \posnum{0.287*} & \posnum{0.133*} \\
& 0.7 & \posnum{0.104} & \posnum{0.188*} & \posnum{0.180*} & \posnum{0.151*} & \posnum{0.165*} & \posnum{0.286*} & \posnum{0.265*} & \posnum{0.168*} & \posnum{0.178*} & \posnum{0.318*} & \posnum{0.312*} & \posnum{0.160*} \\
& 1.0 & \posnum{0.113} & \posnum{0.179*} & \posnum{0.149*} & \posnum{0.156*} & \posnum{0.143} & \posnum{0.223*} & \posnum{0.211*} & \posnum{0.156*} & \posnum{0.156*} & \posnum{0.238*} & \posnum{0.271*} & \posnum{0.139*} \\
\cmidrule{1-14}
\multirow{3}{*}{MS MARCO} & 0.2 & \posnum{0.300*} & \posnum{0.285*} & \posnum{0.293*} & \posnum{0.293*} & \posnum{0.554*} & \posnum{0.535*} & \posnum{0.517*} & \posnum{0.539*} & \posnum{0.674*} & \posnum{0.653*} & \posnum{0.622*} & \posnum{0.655*} \\
& 0.7 & \posnum{0.309*} & \posnum{0.290*} & \posnum{0.298*} & \posnum{0.301*} & \posnum{0.569*} & \posnum{0.550*} & \posnum{0.530*} & \posnum{0.555*} & \posnum{0.693*} & \posnum{0.673*} & \posnum{0.640*} & \posnum{0.676*} \\
& 1.0 & \posnum{0.346*} & \posnum{0.323*} & \posnum{0.334*} & \posnum{0.337*} & \posnum{0.620*} & \posnum{0.586*} & \posnum{0.574*} & \posnum{0.605*} & \posnum{0.749*} & \posnum{0.711*} & \posnum{0.686*} & \posnum{0.731*} \\
\cmidrule(l){2-14}
\multirow{3}{*}{NQ} & 0.2 & \posnum{0.055*} & \posnum{0.056*} & \posnum{0.055*} & \posnum{0.054*} & \posnum{0.134*} & \posnum{0.135*} & \posnum{0.118*} & \posnum{0.131*} & \posnum{0.170*} & \posnum{0.172*} & \posnum{0.145*} & \posnum{0.165*} \\
& 0.7 & \posnum{0.051*} & \posnum{0.056*} & \posnum{0.067*} & \posnum{0.056*} & \posnum{0.157*} & \posnum{0.152*} & \posnum{0.149*} & \posnum{0.158*} & \posnum{0.211*} & \posnum{0.203*} & \posnum{0.188*} & \posnum{0.210*} \\
& 1.0 & \posnum{0.088*} & \posnum{0.080*} & \posnum{0.103*} & \posnum{0.083*} & \posnum{0.218*} & \posnum{0.208*} & \posnum{0.214*} & \posnum{0.213*} & \posnum{0.285*} & \posnum{0.274*} & \posnum{0.267*} & \posnum{0.274*} \\
\cmidrule(l){2-14}
\multirow{3}{*}{HotpotQA} & 0.2 & \posnum{0.111*} & \posnum{0.107*} & \posnum{0.135*} & \posnum{0.108*} & \posnum{0.325*} & \posnum{0.300*} & \posnum{0.337*} & \posnum{0.320*} & \posnum{0.404*} & \posnum{0.370*} & \posnum{0.407*} & \posnum{0.399*} \\
& 0.7 & \posnum{0.135*} & \posnum{0.126*} & \posnum{0.151*} & \posnum{0.133*} & \posnum{0.354*} & \posnum{0.321*} & \posnum{0.358*} & \posnum{0.349*} & \posnum{0.430*} & \posnum{0.387*} & \posnum{0.424*} & \posnum{0.425*} \\
& 1.0 & \posnum{0.149*} & \posnum{0.141*} & \posnum{0.166*} & \posnum{0.146*} & \posnum{0.389*} & \posnum{0.352*} & \posnum{0.391*} & \posnum{0.380*} & \posnum{0.473*} & \posnum{0.423*} & \posnum{0.461*} & \posnum{0.462*} \\
\cmidrule(l){2-14}
\multirow{3}{*}{FEVER} & 0.2 & \posnum{0.033*} & \posnum{0.044*} & \posnum{0.064*} & \posnum{0.035*} & \posnum{0.146*} & \posnum{0.183*} & \posnum{0.184*} & \posnum{0.146*} & \posnum{0.202*} & \posnum{0.250*} & \posnum{0.241*} & \posnum{0.201*} \\
& 0.7 & \posnum{0.036*} & \posnum{0.047*} & \posnum{0.070*} & \posnum{0.034*} & \posnum{0.173*} & \posnum{0.211*} & \posnum{0.216*} & \posnum{0.178*} & \posnum{0.243*} & \posnum{0.290*} & \posnum{0.287*} & \posnum{0.250*} \\
& 1.0 & \posnum{0.052*} & \posnum{0.069*} & \posnum{0.089*} & \posnum{0.056*} & \posnum{0.237*} & \posnum{0.277*} & \posnum{0.267*} & \posnum{0.243*} & \posnum{0.333*} & \posnum{0.379*} & \posnum{0.354*} & \posnum{0.338*} \\
\cmidrule(l){2-14}
\multirow{3}{*}{Climate-FEVER} & 0.2 & \posnum{0.046*} & \posnum{0.103*} & \posnum{0.156*} & \posnum{0.043*} & \posnum{0.063*} & \posnum{0.128*} & \posnum{0.155*} & \posnum{0.059*} & \posnum{0.041*} & \posnum{0.122*} & \posnum{0.133*} & \posnum{0.041*} \\
& 0.7 & \posnum{0.073*} & \posnum{0.139*} & \posnum{0.195*} & \posnum{0.073*} & \posnum{0.095*} & \posnum{0.166*} & \posnum{0.205*} & \posnum{0.105*} & \posnum{0.071*} & \posnum{0.137*} & \posnum{0.181*} & \posnum{0.089*} \\
& 1.0 & \posnum{0.154*} & \posnum{0.218*} & \posnum{0.286*} & \posnum{0.175*} & \posnum{0.248*} & \posnum{0.333*} & \posnum{0.359*} & \posnum{0.263*} & \posnum{0.241*} & \posnum{0.333*} & \posnum{0.331*} & \posnum{0.256*} \\
\cmidrule{1-14}
\multirow{3}{*}{SciDocs} & 0.2 & \posnum{0.024} & \posnum{0.058} & \negnum{-0.004} & \posnum{0.030} & \posnum{0.063} & \posnum{0.114*} & \posnum{0.002} & \posnum{0.061} & \posnum{0.075**} & \posnum{0.133*} & \posnum{0.009} & \posnum{0.058*} \\
& 0.7 & \posnum{0.047} & \posnum{0.076*} & \posnum{0.027} & \posnum{0.044} & \posnum{0.065} & \posnum{0.124*} & \posnum{0.038} & \posnum{0.062} & \posnum{0.041*} & \posnum{0.135*} & \posnum{0.055} & \posnum{0.046*} \\
& 1.0 & \posnum{0.098*} & \posnum{0.106**} & \posnum{0.054} & \posnum{0.100**} & \posnum{0.152**} & \posnum{0.182*} & \posnum{0.083*} & \posnum{0.151*} & \posnum{0.169*} & \posnum{0.206*} & \posnum{0.104**} & \posnum{0.158*} \\
\bottomrule
\end{tabular}
\end{table*}

\section{Qualitative Lexical Differences}
\label{appendix:qualitative_lexical}

This appendix provides qualitative analyses to complement the Zipf’s law findings in the main paper. We examine concrete lexical differences between human-written and LLM-generated texts, illustrating how the observed statistical patterns manifest in actual word choices.

\subsection{Word Usage Pattern Analysis}
\label{appendix:word_changes}

\begin{table*}[t]
\centering
\caption{Words with the largest relative frequency changes between human-written and LLM-generated texts in the MS MARCO dataset.}
\renewcommand{\arraystretch}{1.3}
\resizebox{1\linewidth}{!}{
\begin{tabular}{lcc}
\toprule
Range & Largest Frequency Increases & Largest Frequency Decreases \\
\midrule
Top 2000 & situated, significant, individuals, characterized, typically & important, biggest, mostly, answer, say \\
\addlinespace
Top 2000-10\% & additionally, crucial, furthermore, notably, boasts & wanted, sometime, knows, dies, sort \\
\addlinespace
10\%-20\% & interestingly, noting, encompassing, staggering, firstly & 10000, u.k, catalyzes, tremendous, practically \\
\addlinespace
20\%-50\% & rewritten, boasting, commences, showcasing, delve & dredge, manmade, noncoding, vä, o.s \\
\addlinespace
50\%-100\% & paraphrased, necessitate, chilliest, necessitates, thirdly & humiliated, cubicle, pov, karats, elites \\
\bottomrule
\end{tabular}
}
\label{tab:word_changes}
\end{table*}

This appendix complements our Zipf’s law analysis by examining specific vocabulary changes between human-written and LLM-generated texts. Table~\ref{tab:word_changes} reports the words with the largest relative frequency increases and decreases in LLM-generated texts compared to human-written texts, grouped by different vocabulary frequency ranges.

In the core vocabulary (Top 2000), LLMs systematically favor more formal and precise expressions (``situated", ``significant", ``characterized") over simpler, more casual alternatives (``important", ``biggest", ``say"). This pattern clearly demonstrates LLMs' orientation toward listener's economy, where precise term enhances clarity and reduces cognitive burden for readers.

LLMs also show increased usage of discourse markers and transition words (``additionally'', ``furthermore'', ``notably'', ``interestingly''), indicating greater emphasis on explicit logical structure. In contrast, the reduced usage of highly specialized terms in lower frequency ranges (``noncoding'', ``catalyzes'') and emotional language (``humiliated'') illustrates the more conservative approach to extended vocabulary that our Zipf analysis identified.

These lexical patterns provide concrete examples of how the statistical differences in Zipf parameters manifest in actual word choices, offering additional evidence for our interpretation that LLMs prioritize listener's economy through more diverse and specific core vocabulary while remaining more conservative in extended vocabulary usage.

\subsection{The "Important" Synonym Cluster}

To better understand how LLMs diversify their vocabulary compared to human writers, we conducted a detailed analysis of the ``important'' synonym cluster. We selected this word because it represents a common concept with multiple potential synonyms, making it ideal for studying synonym diversification strategies.

\begin{figure}[t]
    \centering
    \includegraphics[width=0.95\linewidth]{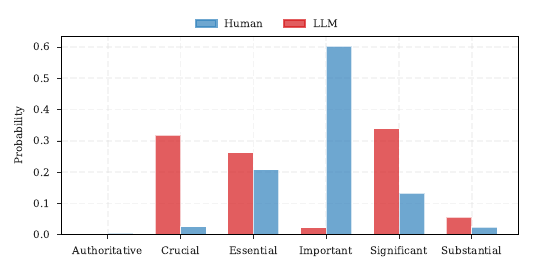}
    \caption{
        Distribution of word usage in the ``important'' synonym cluster for human-written and LLM-generated texts.
    }
    \label{fig:important_synonym_cluster}
\end{figure}

Figure~\ref{fig:important_synonym_cluster} shows the distribution of terms within the "important" synonym cluster. As shown in the figure, human writers tend to default to the general term "important", while LLMs utilize a broader range of contextually more specific alternatives such as "authoritative", "crucial", and "significant". This pattern is consistent with our main findings that LLM-generated texts show higher lexical diversity and more varied synonym usage.

\section{Theoretical Proofs}
\label{appendix:theorem_proof}

This appendix provides the complete mathematical proofs for the distribution alignment theorem as applied to three classical retrieval models: TF-IDF, BM25, and DFR. These proofs substantiate the conceptual arguments presented in the main paper by rigorously showing that, although these models employ different scoring mechanisms, they all implicitly favor documents whose term distributions align closely with the query distribution.

We analyze each model within a consistent analytical framework based on expected retrieval scores under query--document independence. For each model, we derive an upper bound on the expected score and show that this bound is asymptotically maximized when the document-term distribution $P_D(w)$ matches a transformed version of the query-term distribution $P_Q(w)$. This common structure highlights the central role of \emph{distributional alignment} in term-based retrieval models.

Below, we first define the notation used throughout the analysis, then state the shared assumptions, and finally present the detailed proof for each model.

\subsection{Notation}
We now introduce the formal notation used in the following proofs. The key quantities include term frequencies, document and query length statistics, and marginal term distributions over the query and document collections. For a complete overview, see Table~\ref{tab:term_base_notations}.

\begin{table*}[!ht]
\centering
\caption{Description of notations.}
\label{tab:term_base_notations}
\begin{tabular}{l p{0.8\linewidth}} 
\toprule
\textbf{Symbol} & \textbf{Description} \\
\midrule
$V$ & vocabulary. \\
$\mathcal{Q}=\{q_i\}$ & Query set, where $q \sim \mathrm{Unif}(\mathcal{Q})$. \\
$\mathcal{D}=\{d_i\}$ & Document set, where $d \sim \mathrm{Unif}(\mathcal{D})$. \\
$|q|, |d|$ & Length of a query / document. \\
$L_q, L_d$ & Average lengths: $L_q=\mathbb{E}_q[|q|], L_d=\mathbb{E}_d[|d|]$. \\
$\mathrm{tf}(w,q), \mathrm{tf}(w,d)$ & Term frequencies. The expected values are: 
    $\mathbb{E}_q[\mathrm{tf}(w,q)]=L_q P_Q(w)$, 
    $\mathbb{E}_d[\mathrm{tf}(w,d)]=L_d P_D(w)$. \\
$P_Q(w)$ & Probability that a random query token equals $w$. ($\sum_w P_Q(w) = 1$.) \\
$P_D(w)$ & Probability that a random document token equals $w$. ($\sum_w P_D(w) = 1$.) \\
$\mathrm{idf}(w)$ & Pre-computed inverse document frequency, $\log(N/\mathrm{df}(w))$, treated as a constant. \\
\bottomrule
\end{tabular}
\end{table*}

\subsection{Shared Assumptions}

\begin{assumption}[Query–Document Independence]
    \label{as:indep} 
    The random query $q$ is independent of the random document $d$. As a result, the expectation over the joint distribution factors as  
    \begin{equation}
        \mathbb{E}_{q,d}[\,\cdot\,] = \mathbb{E}_q \mathbb{E}_d[\,\cdot\,].
    \end{equation}
\end{assumption}

\begin{assumption}[Bounded Normalized TF Deviation]
    \label{as:dev}
    The variance of normalized term frequencies is uniformly bounded:
    \begin{equation}
        \mathbb{E}_q\left[ \left( \frac{\mathrm{tf}(w,q)}{|q|} - P_Q(w) \right)^2 \right] \le \sigma_q^2, \quad
        \mathbb{E}_d\left[ \left( \frac{\mathrm{tf}(w,d)}{|d|} - P_D(w) \right)^2 \right] \le \sigma_d^2,
    \end{equation}
    for all $w \in V$, where $\sigma_q, \sigma_d \ge 0$ are small constants. 
\end{assumption}

\subsection{TF-IDF}

\begin{proof}
We analyze the expected TF-IDF cosine similarity between a random query $q \in \mathcal{Q}$ and a random document $d \in \mathcal{D}$. The TF-IDF score is:
\begin{equation}
    \mathrm{Score}_{\mathrm{TF\text{-}IDF}}(q,d) = 
    \frac{\sum_{w \in V} t_q(w) \, t_d(w)}
         {\sqrt{\sum_{w \in V} t_q(w)^2} \cdot \sqrt{\sum_{w \in V} t_d(w)^2}},
\end{equation}
where $t_x(w) = \mathrm{tf}(w,x) \cdot \mathrm{idf}(w)$ and $\mathrm{idf}(w) = \log(N / \mathrm{df}(w))$ is treated as a constant.

Define the length-normalized TF--IDF components:
\begin{equation}
    a_w := \frac{t_q(w)}{|q|} = \frac{\mathrm{tf}(w,q)}{|q|} \cdot \mathrm{idf}(w), \quad
    b_w := \frac{t_d(w)}{|d|} = \frac{\mathrm{tf}(w,d)}{|d|} \cdot \mathrm{idf}(w).
\end{equation}
Then the cosine score can be rewritten as:
\begin{equation}
    \mathrm{Score}_{\mathrm{TF\text{-}IDF}}(q,d) = 
    \frac{\sum_{w} a_w b_w \cdot |q| |d|}
         {\left( \sqrt{\sum_{w} a_w^2} \cdot |q| \right)
          \left( \sqrt{\sum_{w} b_w^2} \cdot |d| \right)}
    = \frac{\langle a, b \rangle}{\|a\|_2 \|b\|_2},
\end{equation}
which is the cosine of the normalized TF--IDF vectors $a$ and $b$.

Due to the nonlinearity of the cosine function, $\mathbb{E}_{q,d}[\mathrm{Score}_{\mathrm{TF\text{-}IDF}}(q,d)]$ does not admit a closed-form expression. However, under Assumption~\ref{as:dev}, the normalized term frequencies concentrate around their expectations $P_Q(w)$ and $P_D(w)$ in mean square. This allows us to analyze the expected similarity via a perturbation around the expected vectors.

Define the expected TF--IDF vectors:
\begin{equation}
    \mathbf{x}_w := P_Q(w) \cdot \mathrm{idf}(w), \quad
    \mathbf{y}_w := P_D(w) \cdot \mathrm{idf}(w).
\end{equation}
From Assumption~\ref{as:dev}, we have:
\begin{align}
    \mathbb{E}_q\left[ \left( \frac{\mathrm{tf}(w,q)}{|q|} - P_Q(w) \right)^2 \right] \le \sigma_q^2, \quad \forall w \in V. \\
    \mathbb{E}_d\left[ \left( \frac{\mathrm{tf}(w,d)}{|d|} - P_D(w) \right)^2 \right] \le \sigma_d^2,
    \quad \forall w \in V.
\end{align}
Multiplying by $\mathrm{idf}^2(w)$ and summing over $w$, we obtain the following bounds on the expected squared deviations:
\begin{align}
    \mathbb{E}_q \left[ \|a - \mathbf{x}\|_2^2 \right] \nonumber
    &= \sum_w \mathrm{idf}^2(w) \cdot \mathbb{E}_q\left[ \left( \frac{\mathrm{tf}(w,q)}{|q|} - P_Q(w) \right)^2 \right] \nonumber \\
    & \le \sigma_q^2 \sum_w \mathrm{idf}^2(w) =: \Delta_q, \\
    \mathbb{E}_d \left[ \|b - \mathbf{y}\|_2^2 \right]
    &\le \sigma_d^2 \sum_w \mathrm{idf}^2(w) =: \Delta_d,
\end{align}
where $\Delta_q, \Delta_d$ are small constants when $\sigma_q, \sigma_d$ are small.

Since the cosine function is twice continuously differentiable away from the origin, we expand $\cos(a,b)$ around $(\mathbf{x}, \mathbf{y})$ using a second-order Taylor approximation:
\begin{align}
    \cos(a,b) & = \cos(\mathbf{x},\mathbf{y}) 
    + \nabla_a \cos|_{(\mathbf{x},\mathbf{y})} \cdot (a - \mathbf{x}) \nonumber \\
    & + \nabla_b \cos|_{(\mathbf{x},\mathbf{y})} \cdot (b - \mathbf{y})
    + \mathcal{O}(\|a - \mathbf{x}\|^2 + \|b - \mathbf{y}\|^2).
\end{align}
Taking expectation and noting that $\mathbb{E}[a - \mathbf{x}] = 0$ and $\mathbb{E}[b - \mathbf{y}] = 0$, the first-order terms vanish. Using the bounds on $\mathbb{E}[\|a - \mathbf{x}\|_2^2]$ and $\mathbb{E}[\|b - \mathbf{y}\|_2^2]$, and the fact that the Hessian of cosine is bounded when $\|\mathbf{x}\|_2, \|\mathbf{y}\|_2 \ge \eta > 0$ (which holds since $\mathrm{idf}(w) > 0$ and $P_Q, P_D$ are distributions), we obtain:
\begin{equation}
    \mathbb{E}_{q,d}[\mathrm{Score}_{\mathrm{TF\text{-}IDF}}(q,d)] = \frac{\langle \mathbf{x}, \mathbf{y} \rangle}{\|\mathbf{x}\|_2 \|\mathbf{y}\|_2} + \mathcal{O}(\Delta_q + \Delta_d).
\end{equation}
Thus, the expected similarity is approximated by the cosine of the expected vectors, with approximation error controlled by $\Delta_q$ and $\Delta_d$. When these deviations are small, maximizing the leading term approximately maximizes the expected score.

By the Cauchy--Schwarz inequality:
\begin{equation}
    \langle \mathbf{x}, \mathbf{y} \rangle \le \|\mathbf{x}\|_2 \|\mathbf{y}\|_2,
\end{equation}
with equality if and only if $\mathbf{y} = c \mathbf{x}$ for some $c > 0$. This implies:
\begin{equation}
    P_D(w) \, \mathrm{idf}(w) = c \, P_Q(w) \, \mathrm{idf}(w), \quad \forall w \in V.
\end{equation}
Since $\mathrm{idf}(w) > 0$ for all $w$, we divide both sides to obtain $P_D(w) = c \, P_Q(w)$. Summing over $V$ and using $\sum_w P_D(w) = \sum_w P_Q(w) = 1$, we find $c = 1$. Hence, the unique maximizer is:
\begin{equation}
    P_D^\star(w) = P_Q(w), \quad \forall w \in V.
\end{equation}

Therefore, under bounded deviation of normalized term frequencies, the expected TF--IDF cosine similarity is asymptotically maximized when the document-term distribution matches the query-term distribution.
\end{proof}

\subsection{BM25}

\begin{proof}
We analyze the expected BM25 score between a random query $q \in \mathcal{Q}$ and a random document $d \in \mathcal{D}$. The BM25 score is:
\begin{equation}
    \mathrm{Score}_{\mathrm{BM25}}(q,d) = 
    \sum_{w \in q} \mathrm{idf}(w) \cdot 
    \frac{(k_1 + 1) \cdot \mathrm{tf}(w,d)}{k_1 \left( (1-b) + b \frac{|d|}{\mathrm{avgdl}} \right) + \mathrm{tf}(w,d)},
\end{equation}
where $\mathrm{idf}(w) = \log \frac{N - \mathrm{df}(w) + 0.5}{\mathrm{df}(w) + 0.5}$ , and $k_1 > 0$, $b \in [0,1]$ are hyperparameters.

Define the BM25 term-gain function:
\begin{equation}
    f(t; K) = \frac{(k_1 + 1) \cdot t}{K + t},
\end{equation}
where $K = k_1 \left( (1-b) + b \frac{|d|}{\mathrm{avgdl}} \right)$ is the document-length dependent saturation parameter. Then:
\begin{equation}
    \mathrm{Score}_{\mathrm{BM25}}(q,d) = \sum_{w \in q} \mathrm{idf}(w) \cdot f\bigl(\mathrm{tf}(w,d); K_d\bigr).
\end{equation}

Due to the nonlinearity of $f$, $\mathbb{E}_{q,d}[\mathrm{Score}_{\mathrm{BM25}}]$ does not admit a closed-form expression. However, under Assumption~\ref{as:indep} (query--document independence), we can factor the expectation:
\begin{align}
    \mathbb{E}_{q,d}[\mathrm{Score}_{\mathrm{BM25}}]
    &= \mathbb{E}_{q,d} \left[ \sum_{w \in q} \mathrm{idf}(w) \cdot f\bigl(\mathrm{tf}(w,d); K_d\bigr) \right] \nonumber \\
    &= \sum_w \mathbb{E}_q[\mathbf{1}\{w \in q\}] \cdot \mathrm{idf}(w) \cdot \mathbb{E}_d\left[ f\bigl(\mathrm{tf}(w,d); K_d\bigr) \right] \nonumber \\
    &= \sum_w P_Q(w) \cdot \mathrm{idf}(w) \cdot \mathbb{E}_d\left[ f\bigl(\mathrm{tf}(w,d); K_d\bigr) \right],
    \label{eq:bm25_Escore}
\end{align}
where we used $\mathbb{E}_q[\mathbf{1}\{w \in q\}] = P_Q(w)$.

Now, observe that $f(t; K)$ is strictly increasing and strictly concave in $t$ for fixed $K > 0$, since:
\begin{equation}
    \frac{\partial f}{\partial t} = \frac{(k_1 + 1) K}{(K + t)^2} > 0, \quad
    \frac{\partial^2 f}{\partial t^2} = -\frac{2(k_1 + 1) K}{(K + t)^3} < 0.
\end{equation}

To proceed, we make a mild simplification: assume that documents have typical length, so $K_d \approx k_1$ (when $|d| \approx \mathrm{avgdl}$). This is standard in average-case analysis and does not affect the optimization direction over $P_D$. Under this, $f(\mathrm{tf}(w,d); K_d) \approx f(\mathrm{tf}(w,d); k_1)$, and $f$ becomes a function of $\mathrm{tf}(w,d)$ only.

By Jensen's inequality, and using the concavity of $f$:
\begin{equation}
    \mathbb{E}_d\left[ f\bigl(\mathrm{tf}(w,d); k_1\bigr) \right]
    \le f\left( \mathbb{E}_d[\mathrm{tf}(w,d)]; k_1 \right)
    = f\left( L_d P_D(w); k_1 \right),
\end{equation}
where $L_d = \mathbb{E}_d[|d|]$ and we used $\mathbb{E}_d[\mathrm{tf}(w,d)] = L_d P_D(w)$ from Table~\ref{tab:term_base_notations}.

Substituting into~\eqref{eq:bm25_Escore}:
\begin{equation}
    \mathbb{E}_{q,d}[\mathrm{Score}_{\mathrm{BM25}}]
    \le \sum_w P_Q(w) \, \mathrm{idf}(w) \, f\left( L_d P_D(w); k_1 \right)
    =: F(P_D).
    \label{eq:bm25_functional}
\end{equation}

The functional $F(P_D)$ is strictly concave in $P_D$ because:
- $f(\cdot; k_1)$ is strictly concave,
- $P_Q(w) \mathrm{idf}(w) \ge 0$, and positive for $w$ in typical queries,
- hence each term is strictly concave in $P_D(w)$.

Therefore, $F(P_D)$ has a unique maximizer $P_D^\star$ on the probability simplex $\{ P_D \mid P_D(w) \ge 0, \sum_w P_D(w) = 1 \}$.

To find $P_D^\star$, introduce a Lagrange multiplier $\lambda$ for the constraint $\sum_w P_D(w) = 1$:
\begin{equation}
    \mathcal{L} = F(P_D) - \lambda \left( \sum_w P_D(w) - 1 \right).
\end{equation}

The first-order condition for $P_D^\star(w) > 0$ is:
\begin{equation}
    \frac{\partial \mathcal{L}}{\partial P_D(w)}
    = P_Q(w) \, \mathrm{idf}(w) \cdot L_d \cdot \frac{\partial f}{\partial t}\bigg|_{t = L_d P_D^\star(w)}
    - \lambda = 0.
\end{equation}

Compute the derivative:
\begin{equation}
    \frac{\partial f}{\partial t} = \frac{(k_1 + 1) k_1}{(k_1 + t)^2},
\end{equation}
so:
\begin{equation}
    P_Q(w) \, \mathrm{idf}(w) \cdot L_d \cdot \frac{(k_1 + 1) k_1}{(k_1 + L_d P_D^\star(w))^2} = \lambda.
\end{equation}

Since $\lambda$ is independent of $w$, we can write:
\begin{equation}
    (k_1 + L_d P_D^\star(w))^2 = c \cdot P_Q(w) \, \mathrm{idf}(w),
    \quad c > 0 \text{ (constant)}.
\end{equation}

Solving for $P_D^\star(w)$:
\begin{equation}
    P_D^\star(w) = \max\left\{ 0,\ \frac{1}{L_d} \left( \sqrt{c \, P_Q(w) \, \mathrm{idf}(w)} - k_1 \right) \right\}.
\end{equation}

Let $\alpha = \sqrt{c}/L_d$, then:
\begin{equation}
    P_D^\star(w) = \max\left\{ 0,\ \alpha \sqrt{P_Q(w) \, \mathrm{idf}(w)} - \frac{k_1}{L_d} \right\}.
    \label{eq:bm25_solution}
\end{equation}

The constant $\alpha > 0$ is determined by normalization $\sum_w P_D^\star(w) = 1$, and is unique due to strict concavity.

In standard retrieval settings, the average document length $L_d$ is significantly larger than the saturation parameter $k_1$, so $k_1 / L_d \ll 1$. This makes the constant offset in~\eqref{eq:bm25_solution} negligible in practice, leading to the approximation:
\begin{equation}
    P_D^\star(w) \propto \sqrt{P_Q(w) \, \mathrm{idf}(w)}.
    \label{eq:bm25_root}
\end{equation}

Consequently, under bounded term frequency fluctuations and typical document lengths, the expected BM25 score is asymptotically maximized when the document-term distribution aligns with the geometric mean of $P_Q(w)$ and $\mathrm{idf}(w)$.
\end{proof}

Empirical analyses of large query logs have shown a weak but positive correlation between $\mathrm{idf}(w)$ and $P_Q(w)$, reflecting users’ tendency to select discriminative, information-rich terms~\cite{jansen2000real,craswellEffective2001, arampatzisEmpirical2010}. 
Motivated by this observation, we consider an \emph{idealized strong assumption} under which $\mathrm{idf}(w) \propto P_Q(w)$. 
In this regime, the square root in~\eqref{eq:bm25_root} becomes
\[
    \sqrt{P_Q(w) \cdot \mathrm{idf}(w)} \propto P_Q(w),
\]
and the optimal document distribution collapses to
\begin{equation}
    P_D^\star(w) \propto P_Q(w).
\end{equation}
Combined with the normalization constraint $\sum_w P_D^\star(w) = 1$, this yields $P_D^\star = P_Q$. Thus, in this regime, BM25's optimal alignment condition coincides with that of Query Likelihood and TF-IDF.

\subsection{DFR}

\begin{proof}
We analyze the expected DFR score between a random query $q \in \mathcal{Q}$ and a random document $d \in \mathcal{D}$. The Divergence From Randomness (DFR) framework models term weight as a product of two components:
\begin{equation}
    \text{Score}_{\text{DFR}}(q,d) = \sum_{w \in q} \underbrace{[-\log P_1(w|d)]}_{\text{information gain from frequency}} \cdot \underbrace{[1 - P_2(w|d)]}_{\text{normalization by elite set}},
\end{equation}
where $P_1(w|d)$ measures the likelihood that the term frequency $\mathrm{tf}(w,d)$ is generated by a random process, and $P_2(w|d)$ adjusts for document length or importance via an elite set. 
We focus on the widely used \textsc{InL2} variant of the DFR framework. In the \textsc{InL2} model, the term weight is defined as:
\begin{align}
    -\log P_1(w|d) &= \mathrm{tf}(w,d) \cdot \log_2 \frac{N + 1}{\mathrm{df}(w) + 0.5}, \\
    1 - P_2(w|d) &= \frac{1}{\mathrm{tf}(w,d) + 1},
\end{align}
following the standard instantiation of the DFR framework.
This yields the term weight:
\begin{equation}
    \mathrm{Score}_{\mathrm{DFR}}(q,d) = \sum_{w \in q} 
    \left[ \mathrm{tf}(w,d) \cdot \log_2 \frac{N + 1}{\mathrm{df}(w) + 0.5} \right] \cdot
    \frac{1}{\mathrm{tf}(w,d) + 1},
\end{equation}
which simplifies to:
\begin{equation}
    \mathrm{Score}_{\mathrm{DFR}}(q,d) = \sum_{w \in q} \mathrm{idf}(w) \cdot f\bigl(\mathrm{tf}(w,d)\bigr),
\end{equation}
where $\mathrm{idf}(w) = \log_2 \frac{N + 1}{\mathrm{df}(w) + 0.5}$, and $f(t) = \frac{t}{t + 1}$ is a function that is strictly increasing and strictly concave.

Then:
\begin{equation}
    \mathbb{E}_{q,d}[\mathrm{Score}_{\mathrm{DFR}}] = \sum_w P_Q(w) \cdot \mathrm{idf}(w) \cdot \mathbb{E}_d\left[ f\bigl(\mathrm{tf}(w,d)\bigr) \right].
\end{equation}

Since $f$ is concave, by Jensen’s inequality:
\begin{equation}
    \mathbb{E}_d\left[ f\bigl(\mathrm{tf}(w,d)\bigr) \right] \le f\left( \mathbb{E}_d[\mathrm{tf}(w,d)] \right) = \frac{L_d P_D(w)}{L_d P_D(w) + 1}.
\end{equation}

Thus:
\begin{equation}
    \mathbb{E}_{q,d}[\mathrm{Score}_{\mathrm{DFR}}] \le \sum_w P_Q(w) \, \mathrm{idf}(w) \cdot \frac{L_d P_D(w)}{L_d P_D(w) + 1} =: F(P_D).
\end{equation}

$F(P_D)$ is strictly concave in $P_D$, so it has a unique maximizer $P_D^\star$.

Using Lagrange multiplier $\lambda$:
\begin{equation}
    \frac{\partial \mathcal{L}}{\partial P_D(w)} = P_Q(w) \, \mathrm{idf}(w) \cdot L_d \cdot \frac{1}{(L_d P_D^\star(w) + 1)^2} - \lambda = 0.
\end{equation}

This implies:
\begin{equation}
    (L_d P_D^\star(w) + 1)^2 = c \cdot P_Q(w) \, \mathrm{idf}(w),
\end{equation}
so:
\begin{equation}
    P_D^\star(w) = \max\left\{ 0,\ \alpha \sqrt{P_Q(w) \, \mathrm{idf}(w)} - \frac{1}{L_d} \right\}, \quad \alpha > 0.
\end{equation}

For $L_d \gg 1$, $1 / L_d \ll 1$, so:
\begin{equation}
    P_D^\star(w) \propto \sqrt{P_Q(w) \, \mathrm{idf}(w)}.
    \label{eq:dfr_root}
\end{equation}

Therefore, the expected DFR score is asymptotically maximized when $P_D(w)$ is proportional to the geometric mean of $P_Q(w)$ and $\mathrm{idf}(w)$.
\end{proof}

As shown in the analysis of BM25, under the empirically observed positive correlation between $\mathrm{idf}(w)$ and $P_Q(w)$, and particularly when $\mathrm{idf}(w) \propto P_Q(w)$, the optimal document distribution for DFR also converges to $P_D^\star = P_Q$. 
Thus, in this regime, the alignment conditions of DFR, BM25, and TF-IDF are identical.

This reveals a deep connection: although DFR, BM25, and TF-IDF are based on different scoring mechanisms—divergence from randomness, probabilistic ranking, and vector space geometry, respectively—they all implicitly encourage \emph{distributional alignment} between the query and document collections. The specific form of alignment (exact, smoothed, or IDF-weighted) depends on each model’s assumptions and inductive biases, but the underlying principle remains the same: \textbf{effective retrieval arises from aligning the query’s term distribution with that of the document}.

\section{Experiment Settings and Extensions}
\label{appendix:experiment_setting_and_extension}

\subsection{Retrieval Parameters}
\label{appendix:main_experiment_retrieval_config}
This section details the parameter configurations used in all retrieval methods. All main experiments in the paper employ the LLaMA-generated corpus with temperature $0.7$, consistent with the preliminary experiments in Appendix~\ref{appendix:preliminary_experiments}, ensuring methodological coherence.

The retrieval settings are as follows:

\begin{itemize}
    \item \textbf{TF-IDF}: Standard logarithmic TF-IDF with document length normalization
    \item \textbf{BM25}: $k_1=0.9$, $b=0.4$ \cite{yang2018anserini}
    \item \textbf{Query Likelihood (QL)}: Jelinek-Mercer smoothing, $\lambda=0.1$
    \item \textbf{DFR}: In-L-H2 model, $c=1.0$, $z=0.3$
\end{itemize}

All implementations were based on the Pyserini toolkit \cite{linPyserini2021} with custom extensions for DFR and TF-IDF functionalities.

\subsection{Porter Stemming and Stopwords Removal}
\label{appendix:stemming_stopwords}

The main experiments were conducted without stemming or stopwords removal to maintain consistency with our linguistic analysis. Here we present additional experiments using Porter stemming with stopwords removal, which represents standard practice in term-based retrieval.

Table~\ref{tab:preference_analysis_porter} shows the complete results with this standard preprocessing configuration. The pattern remains consistent: Human queries show a strong preference for human-written documents, while LLM queries often prefer LLM-generated content or show a reduced preference for human content. This stability across preprocessing methods further validates that our findings regarding distribution alignment and source preferences are robust and generalizable.

\begin{table*}[t]
\centering
\caption{Preference analysis results across different datasets and retrieval methods \textbf{Porter stemming and stopwords removal} (top 200 per query).  $\Delta$SR@10, $\Delta$NDSR@10, and $\Delta$MASR represent the preference metrics based on SR@10, NDSR@10, and MASR respectively. \textcolor{humanblue}{blue cells} indicate preference towards human-written documents, while \textcolor{llmred}{red cells} indicate preference towards LLM-generated content. * indicates statistical significance ($p < 0.05$).}
\label{tab:preference_analysis_porter}
\aboverulesep=0pt
\belowrulesep=0pt
\setlength{\heavyrulewidth}{0.5pt}
\setlength{\lightrulewidth}{0.5pt}
\resizebox{1\linewidth}{!}{
\renewcommand{\arraystretch}{1.3}
\begin{tabular}{cccccccccccccc}
\toprule
\multirow{2}{*}{\centering Dataset} & \multirow{2}{*}{\centering Query Type} & \multicolumn{4}{c}{$\Delta$SR@10} & \multicolumn{4}{c}{$\Delta$NDSR@10} & \multicolumn{4}{c}{$\Delta$MASR} \\
\cmidrule(lr){3-6} \cmidrule(lr){7-10} \cmidrule(lr){11-14}
& & TF-IDF & BM25 & QL & DFR & TF-IDF & BM25 & QL & DFR & TF-IDF & BM25 & QL & DFR \\
\midrule
\multirow{2}{*}{ANTIQUE} & Human & \posnum{0.428*} & \posnum{0.382*} & \posnum{0.450*} & \posnum{0.424*} & \posnum{0.460*} & \posnum{0.414*} & \posnum{0.474*} & \posnum{0.456*} & \posnum{0.376*} & \posnum{0.349*} & \posnum{0.412*} & \posnum{0.371*} \\
& LLM & \posnum{0.008} & \negnum{-0.025*} & \negnum{-0.006} & \negnum{-0.005} & \posnum{0.014} & \negnum{-0.025*} & \negnum{-0.013} & \negnum{-0.001} & \posnum{0.035*} & \posnum{0.011} & \posnum{0.039*} & \posnum{0.024*} \\
\cmidrule{2-14}
\multirow{2}{*}{FiQA} & Human & \posnum{0.080*} & \posnum{0.201*} & \posnum{0.183*} & \posnum{0.096*} & \posnum{0.090*} & \posnum{0.217*} & \posnum{0.188*} & \posnum{0.101*} & \posnum{0.072*} & \posnum{0.172*} & \posnum{0.179*} & \posnum{0.084*} \\
& LLM & \negnum{-0.153*} & \negnum{-0.018} & \negnum{-0.079*} & \negnum{-0.151*} & \negnum{-0.160*} & \negnum{-0.022} & \negnum{-0.086*} & \negnum{-0.156*} & \negnum{-0.113*} & \negnum{-0.004} & \negnum{-0.027*} & \negnum{-0.104*} \\
\midrule
\multirow{2}{*}{MS MARCO} & Human & \posnum{0.411*} & \posnum{0.390*} & \posnum{0.364*} & \posnum{0.396*} & \posnum{0.449*} & \posnum{0.428*} & \posnum{0.396*} & \posnum{0.434*} & \posnum{0.336*} & \posnum{0.322*} & \posnum{0.310*} & \posnum{0.327*} \\
& LLM & \posnum{0.249*} & \posnum{0.222*} & \posnum{0.128*} & \posnum{0.228*} & \posnum{0.272*} & \posnum{0.243*} & \posnum{0.143*} & \posnum{0.250*} & \posnum{0.199*} & \posnum{0.181*} & \posnum{0.120*} & \posnum{0.184*} \\
\cmidrule{2-14}
\multirow{2}{*}{NQ} & Human & \posnum{0.122*} & \posnum{0.147*} & \posnum{0.155*} & \posnum{0.122*} & \posnum{0.138*} & \posnum{0.158*} & \posnum{0.161*} & \posnum{0.137*} & \posnum{0.117*} & \posnum{0.141*} & \posnum{0.150*} & \posnum{0.118*} \\
& LLM & \posnum{0.017*} & \posnum{0.027*} & \negnum{-0.009} & \posnum{0.011*} & \posnum{0.023*} & \posnum{0.031*} & \negnum{-0.011} & \posnum{0.015*} & \posnum{0.016*} & \posnum{0.036*} & \posnum{0.011*} & \posnum{0.013*} \\
\cmidrule{2-14}
\multirow{2}{*}{HotpotQA} & Human & \posnum{0.149*} & \posnum{0.121*} & \posnum{0.100*} & \posnum{0.140*} & \posnum{0.185*} & \posnum{0.156*} & \posnum{0.140*} & \posnum{0.177*} & \posnum{0.112*} & \posnum{0.092*} & \posnum{0.068*} & \posnum{0.104*} \\
& LLM & \posnum{0.007*} & \negnum{-0.018*} & \negnum{-0.081*} & \negnum{-0.005} & \posnum{0.028*} & \posnum{0.001} & \negnum{-0.061*} & \posnum{0.016*} & \negnum{-0.026*} & \negnum{-0.042*} & \negnum{-0.097*} & \negnum{-0.036*} \\
\cmidrule{2-14}
\multirow{2}{*}{FEVER} & Human & \posnum{0.071*} & \posnum{0.079*} & \posnum{0.074*} & \posnum{0.063*} & \posnum{0.097*} & \posnum{0.104*} & \posnum{0.095*} & \posnum{0.090*} & \posnum{0.060*} & \posnum{0.070*} & \posnum{0.068*} & \posnum{0.056*} \\
& LLM & \negnum{-0.019*} & \posnum{0.000} & \negnum{-0.040*} & \negnum{-0.026*} & \negnum{-0.007*} & \posnum{0.012*} & \negnum{-0.040*} & \negnum{-0.014*} & \negnum{-0.038*} & \negnum{-0.014*} & \negnum{-0.049*} & \negnum{-0.043*} \\
\cmidrule{2-14}
\multirow{2}{*}{Climate-FEVER} & Human & \posnum{0.031*} & \posnum{0.096*} & \posnum{0.155*} & \posnum{0.032*} & \posnum{0.033*} & \posnum{0.102*} & \posnum{0.163*} & \posnum{0.035*} & \posnum{0.036*} & \posnum{0.090*} & \posnum{0.130*} & \posnum{0.041*} \\
& LLM & \negnum{-0.060*} & \posnum{0.003} & \posnum{0.036*} & \negnum{-0.063*} & \negnum{-0.064*} & \negnum{-0.000} & \posnum{0.036*} & \negnum{-0.067*} & \negnum{-0.036*} & \posnum{0.017*} & \posnum{0.036*} & \negnum{-0.033*} \\
\midrule
\multirow{2}{*}{SciDocs} & Human & \posnum{0.042*} & \posnum{0.076*} & \posnum{0.044*} & \posnum{0.039*} & \posnum{0.044*} & \posnum{0.084*} & \posnum{0.035*} & \posnum{0.038*} & \posnum{0.018*} & \posnum{0.046*} & \posnum{0.017*} & \posnum{0.018*} \\
& LLM & \negnum{-0.064*} & \negnum{-0.038*} & \negnum{-0.116*} & \negnum{-0.073*} & \negnum{-0.074*} & \negnum{-0.043*} & \negnum{-0.133*} & \negnum{-0.083*} & \negnum{-0.053*} & \negnum{-0.024*} & \negnum{-0.085*} & \negnum{-0.055*} \\
\bottomrule
\end{tabular}
}
\end{table*}

\begin{table*}[t]
\centering
\caption{Preference analysis results across different datasets and retrieval methods \textbf{without stemming and keeping stopwords} (top 200 per query).  $\Delta$SR@10, $\Delta$NDSR@10, and $\Delta$MASR represent the preference metrics based on SR@10, NDSR@10, and MASR respectively. \textcolor{humanblue}{blue cells} indicate preference towards human-written documents, while \textcolor{llmred}{red cells} indicate preference towards LLM-generated content. * indicates statistical significance ($p < 0.05$).}
\label{tab:preference_analysis_no_stem}
\renewcommand{\arraystretch}{1.3}
\resizebox{1\linewidth}{!}{
\begin{tabular}{cccccccccccccc}
\toprule
\multirow{2}{*}{\centering Dataset} & \multirow{2}{*}{\centering Query Type} & \multicolumn{4}{c}{$\Delta$SR@10} & \multicolumn{4}{c}{$\Delta$NDSR@10} & \multicolumn{4}{c}{$\Delta$MASR} \\
\cmidrule(lr){3-6} \cmidrule(lr){7-10} \cmidrule(lr){11-14}
& & TF-IDF & BM25 & QL & DFR & TF-IDF & BM25 & QL & DFR & TF-IDF & BM25 & QL & DFR \\
\midrule
\multirow{2}{*}{ANTIQUE} & Human & \posnum{0.436*} & \posnum{0.400*} & \posnum{0.449*} & \posnum{0.427*} & \posnum{0.463*} & \posnum{0.421*} & \posnum{0.469*} & \posnum{0.454*} & \posnum{0.381*} & \posnum{0.366*} & \posnum{0.419*} & \posnum{0.376*} \\
& LLM & \negnum{-0.017} & \negnum{-0.028*} & \negnum{-0.018} & \negnum{-0.031*} & \negnum{-0.015} & \negnum{-0.033*} & \negnum{-0.028*} & \negnum{-0.030*} & \posnum{0.018*} & \posnum{0.014*} & \posnum{0.040*} & \posnum{0.008} \\
\cmidrule{2-14}
\multirow{2}{*}{FiQA} & Human & \posnum{0.053*} & \posnum{0.190*} & \posnum{0.195*} & \posnum{0.062*} & \posnum{0.059*} & \posnum{0.201*} & \posnum{0.197*} & \posnum{0.069*} & \posnum{0.069*} & \posnum{0.179*} & \posnum{0.196*} & \posnum{0.082*} \\
& LLM & \negnum{-0.173*} & \negnum{-0.027} & \negnum{-0.039*} & \negnum{-0.165*} & \negnum{-0.184*} & \negnum{-0.034*} & \negnum{-0.054*} & \negnum{-0.174*} & \negnum{-0.110*} & \posnum{0.012} & \posnum{0.015} & \negnum{-0.099*} \\
\midrule
\multirow{2}{*}{MS MARCO} & Human & \posnum{0.433*} & \posnum{0.408*} & \posnum{0.391*} & \posnum{0.418*} & \posnum{0.464*} & \posnum{0.440*} & \posnum{0.413*} & \posnum{0.450*} & \posnum{0.361*} & \posnum{0.346*} & \posnum{0.342*} & \posnum{0.351*} \\
& LLM & \posnum{0.262*} & \posnum{0.231*} & \posnum{0.145*} & \posnum{0.241*} & \posnum{0.283*} & \posnum{0.250*} & \posnum{0.152*} & \posnum{0.261*} & \posnum{0.213*} & \posnum{0.194*} & \posnum{0.140*} & \posnum{0.198*} \\
\cmidrule{2-14}
\multirow{2}{*}{NQ} & Human & \posnum{0.125*} & \posnum{0.155*} & \posnum{0.179*} & \posnum{0.123*} & \posnum{0.138*} & \posnum{0.164*} & \posnum{0.184*} & \posnum{0.136*} & \posnum{0.125*} & \posnum{0.157*} & \posnum{0.183*} & \posnum{0.127*} \\
& LLM & \posnum{0.019*} & \posnum{0.043*} & \posnum{0.032*} & \posnum{0.014*} & \posnum{0.027*} & \posnum{0.048*} & \posnum{0.028*} & \posnum{0.020*} & \posnum{0.027*} & \posnum{0.055*} & \posnum{0.056*} & \posnum{0.025*} \\
\cmidrule{2-14}
\multirow{2}{*}{HotpotQA} & Human & \posnum{0.152*} & \posnum{0.138*} & \posnum{0.146*} & \posnum{0.144*} & \posnum{0.186*} & \posnum{0.170*} & \posnum{0.177*} & \posnum{0.179*} & \posnum{0.124*} & \posnum{0.118*} & \posnum{0.121*} & \posnum{0.118*} \\
& LLM & \posnum{0.005} & \negnum{-0.003} & \negnum{-0.039*} & \negnum{-0.002} & \posnum{0.029*} & \posnum{0.014} & \negnum{-0.024*} & \posnum{0.021*} & \negnum{-0.015*} & \negnum{-0.016*} & \negnum{-0.042*} & \negnum{-0.022*} \\
\cmidrule{2-14}
\multirow{2}{*}{FEVER} & Human & \posnum{0.076*} & \posnum{0.088*} & \posnum{0.101*} & \posnum{0.071*} & \posnum{0.097*} & \posnum{0.112*} & \posnum{0.122*} & \posnum{0.094*} & \posnum{0.068*} & \posnum{0.087*} & \posnum{0.099*} & \posnum{0.065*} \\
& LLM & \negnum{-0.025*} & \posnum{0.007*} & \negnum{-0.018*} & \negnum{-0.028*} & \negnum{-0.012*} & \posnum{0.018*} & \negnum{-0.015*} & \negnum{-0.016*} & \negnum{-0.037*} & \negnum{-0.002} & \negnum{-0.021*} & \negnum{-0.040*} \\
\cmidrule{2-14}
\multirow{2}{*}{Climate-FEVER} & Human & \posnum{0.035*} & \posnum{0.124*} & \posnum{0.215*} & \posnum{0.043*} & \posnum{0.036*} & \posnum{0.130*} & \posnum{0.227*} & \posnum{0.044*} & \posnum{0.047*} & \posnum{0.120*} & \posnum{0.191*} & \posnum{0.056*} \\
& LLM & \negnum{-0.068*} & \posnum{0.016*} & \posnum{0.077*} & \negnum{-0.059*} & \negnum{-0.079*} & \posnum{0.014} & \posnum{0.077*} & \negnum{-0.070*} & \negnum{-0.033*} & \posnum{0.040*} & \posnum{0.088*} & \negnum{-0.026*} \\
\midrule
\multirow{2}{*}{SciDocs} & Human & \posnum{0.010} & \posnum{0.053*} & \posnum{0.032*} & \posnum{0.014*} & \posnum{0.012} & \posnum{0.064*} & \posnum{0.029*} & \posnum{0.016*} & \posnum{0.005} & \posnum{0.042*} & \posnum{0.024*} & \posnum{0.008*} \\
& LLM & \negnum{-0.096*} & \negnum{-0.050*} & \negnum{-0.109*} & \negnum{-0.097*} & \negnum{-0.106*} & \negnum{-0.056*} & \negnum{-0.124*} & \negnum{-0.108*} & \negnum{-0.073*} & \negnum{-0.035*} & \negnum{-0.076*} & \negnum{-0.073*} \\
\bottomrule
\end{tabular}
}
\end{table*}

\begin{figure}[t]
    \centering
    \includegraphics[width=\linewidth]{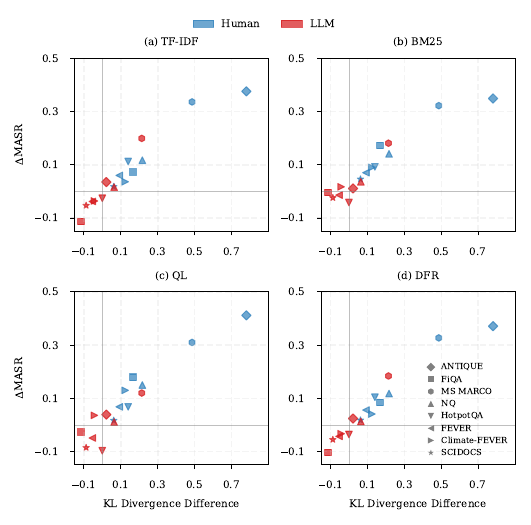}
    \caption{
        KL-divergence differences ($\text{KL}_{\text{LLaMA}} - \text{KL}_{\text{human}}$) versus $\Delta$ MASR with 
        \textbf{Porter stemming and stopwords removal} for four retrieval methods: (a) TF-IDF, (b) BM25, (c) QL, and (d) DFR. Blue and red points represent human and LLM queries respectively, with each shape denoting a different dataset.
    }
    \label{fig:masr_comparison_porter}
\end{figure}

To further examine the relationship between distribution alignment and retrieval preferences, Figure~\ref{fig:masr_comparison_porter} shows the correlation between KL-divergence differences and $\Delta$MASR with Porter stemming. We observe high Pearson correlation coefficients of 0.96, 0.95, 0.92, and 0.97 for TF-IDF, BM25, QL, and DFR respectively, all statistically significant (p < 0.001). These strong correlations across all four retrieval methods confirm that our distribution alignment theorem remains robust across different preprocessing configurations.

\subsection{Metrics Extension}
\label{appendix:main_metrics_extension}
 To further validate our distribution alignment theorem, we extend our analysis to three evaluation metrics for the two preprocessing methods.

 The following tables summarize the performance of the three metrics under each preprocessing configuration: Table~\ref{tab:preference_analysis_porter} shows the results with stemming and stopwords removal, while Table~\ref{tab:preference_analysis_no_stem} shows the results without stemming and keeping stopwords.

\begin{figure*}[t]
\centering
\includegraphics[width=\linewidth]{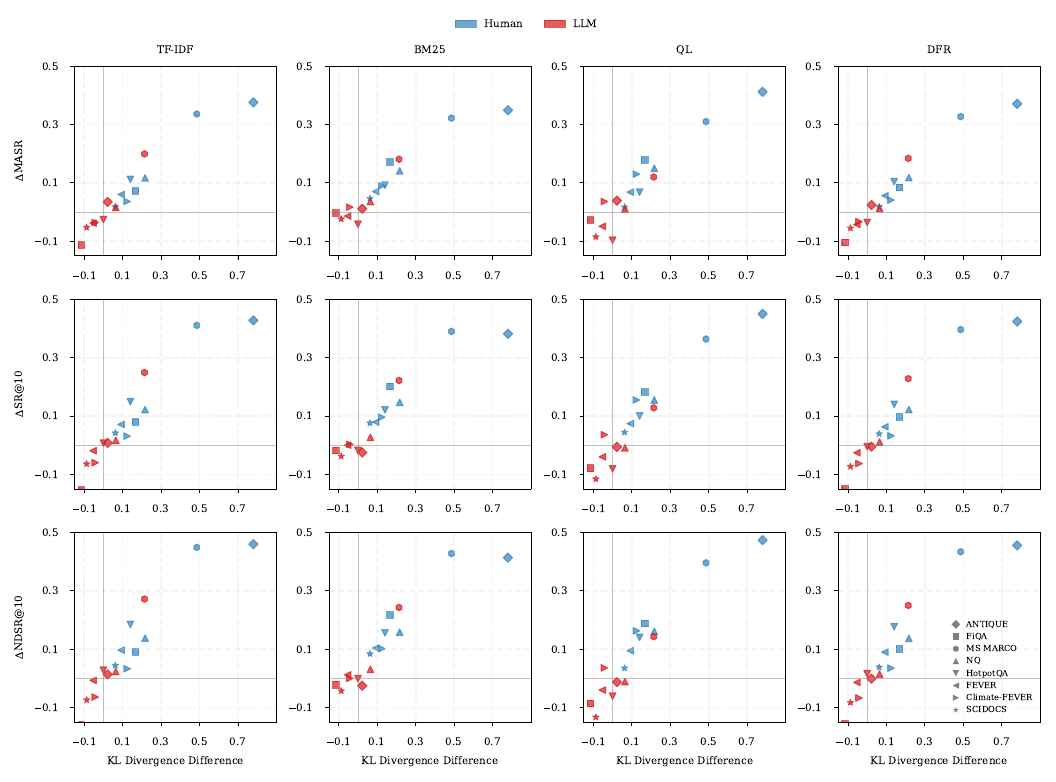}
\caption{
KL-divergence differences versus preference metrics with \textbf{Porter stemming and stopwords removal} for four retrieval methods. Rows correspond to $\Delta$SR@10, $\Delta$NDSR@10, and $\Delta$MASR; columns correspond to TF-IDF, BM25, QL, and DFR. Blue points represent human queries, red points represent LLM queries, with each point representing a dataset.
}
\label{fig:kl_divergence_all_metrics_stem}
\end{figure*}

Figure~\ref{fig:kl_divergence_all_metrics_stem} demonstrates that the strong correlation between KL-divergence differences and retrieval preferences holds consistently across multiple metrics with Porter stemming. This pattern extends to both Source Ratio (SR@10) and Normalized Discounted Source Ratio (NDSR@10), shown in the top and middle rows. All correlations remain statistically significant (p < 0.05) across all four retrieval methods.
\begin{figure*}[t]
\centering
\includegraphics[width=\linewidth]{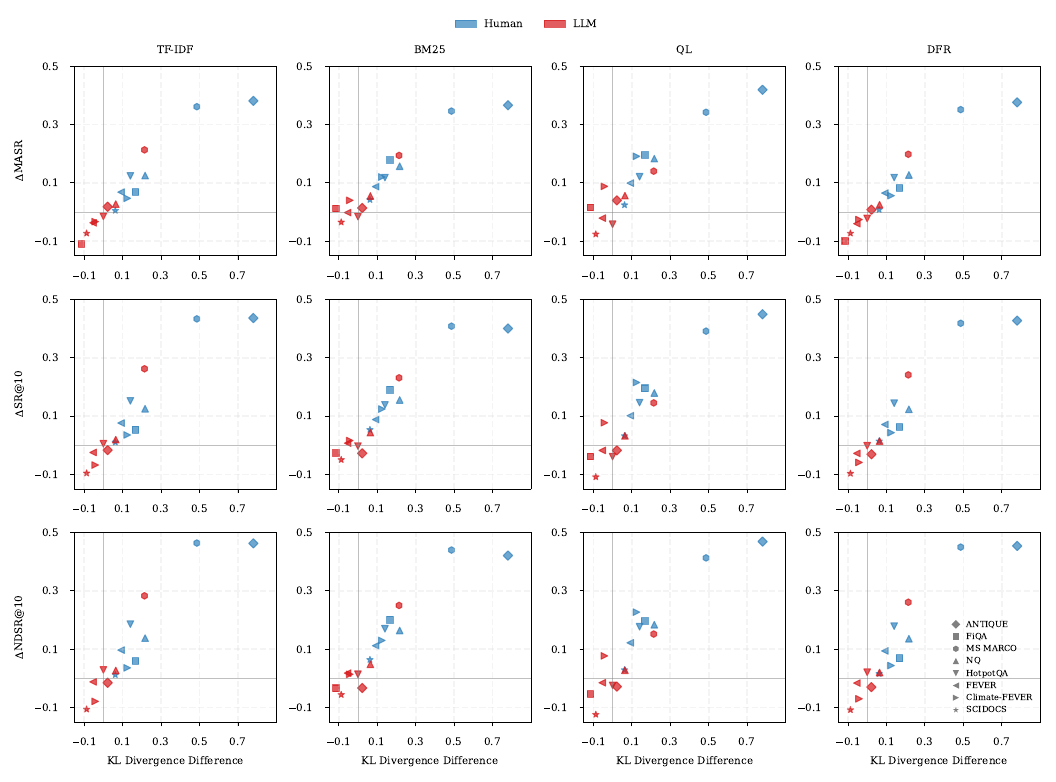}
\caption{
KL-divergence differences versus preference metrics \textbf{without stemming and keeping stopwords} for four retrieval methods. Rows correspond to $\Delta$SR@10, $\Delta$NDSR@10, and $\Delta$MASR; columns correspond to TF-IDF, BM25, QL, and DFR. Blue points represent human queries, red points represent LLM queries, with each point representing a dataset.
}
\label{fig:kl_divergence_all_metrics_no_stem}
\end{figure*}

For completeness, we also examine configurations without stemming, as shown in Figure~\ref{fig:kl_divergence_all_metrics_no_stem}. The strong correlations persist across all metrics and retrieval methods (all p < 0.05). This consistency further validates that our findings are not artifacts of specific preprocessing choices.

The robust relationship across different metrics and preprocessing settings underscores our conclusion that term-based retrieval models favor documents with term distributions that better match the query's distribution, rather than exhibiting inherent source bias. Importantly, this relationship holds regardless of whether we measure retrieval effectiveness using SR@10, NDSR@10, or MASR, suggesting that the distribution alignment principle is fundamental to understanding term-based retrieval behavior across various effectiveness measures.

\end{document}